\documentclass[11pt]{article}
\usepackage{graphicx,amssymb,amsmath,amsthm,natbib,setspace}

\newtheorem{remark}{Remark}[section]
\newtheorem{cor}{Corollary}[section]

\newtheorem{prop}{Proposition}[section]

\numberwithin{equation}{section}

\newcommand{\E}{\mathbb{E}}

\newcommand{\eps}{\epsilon}
\renewcommand{\P}{\mathbb{P}}
\renewcommand{\H}{\mathcal{H}}
\newcommand{\F}{\mathcal{F}}
\newcommand{\la}{\lambda}
\newcommand{\G}{\mathcal{G}}
\renewcommand{\L}{\mathcal{L}}
\newcommand{\M}{\mathcal{M}}

\title{\LARGE \bf
\sc A Unified Framework for Pricing Credit and Equity Derivatives \thanks{Department of Mathematics, University
of Michigan, Ann Arbor, MI 48109, USA, e-mail:\{erhan,boy\}@umich.edu.} \thanks{We would like to thank the two anonymous referees and the anonymous AE for their constructive comments, which helped us improve our paper.} }

\author{Erhan Bayraktar % <-this % stops a space
\thanks{This
work is supported in part by the National Science Foundation,
under grant DMS-060449.} \and Bo Yang}
\date{}

\begin{document}
\maketitle

\begin{abstract}
We propose a model which can be jointly calibrated to the corporate bond term structure and equity option volatility surface of the same company. Our purpose is to obtain explicit bond and equity option pricing formulas that can be calibrated 
to find a risk neutral model that matches a set of observed market prices. This risk neutral model can then be used to price more exotic, illiquid or over-the-counter derivatives. 
We observe that the model implied credit default swap (CDS) spread matches the market CDS spread and that our model produces a very desirable CDS spread term structure. This is observation is worth noticing since without calibrating any parameter to the CDS spread data, it is matched by the CDS spread that our model generates using the available information from the equity options and corporate bond markets.
We also observe that our model matches the equity option implied volatility surface well since we properly account for the default risk premium in the implied volatility surface.
We demonstrate the importance of accounting for the default risk and stochastic interest rate in equity option pricing by comparing our results to \cite{ronnie-timescale}, which only accounts for stochastic volatility. 

\end{abstract}

\tableofcontents

\section{Introduction}

Our purpose is to build an intensity-based modeling framework that can be used in trading and calibrating across the credit and equity markets. The same company has  stocks, stock options, bonds, credit default swaps on these bonds, and several other derivatives. When this company defaults, the payoffs of all of these instruments are affected; therefore, their prices all contain information about the default risk of the company. 

We build a model that can be jointly calibrated to corporate bond prices and stock options, and can be used to price more exotic derivatives. In our framework we use the Vasicek model for the interest rate, and use doubly stochastic Poisson process to model the default of a given company. We assume that the bonds have recovery of market value and that stocks become valueless at the time of default.
Using the multi-scale modeling approach of \cite{ronnie-timescale}  we obtain explicit bond pricing equation with three free parameters which we calibrate to the corporate bond term structure. On the other hand, stock option pricing formula contain seven parameters, three of which are common with the bond option pricing formula. (The common parameters are multiplied with the loss rate in the bond pricing formula.) 
We calibrate the remaining set of parameters to the stock option prices. This hybrid model, therefore, is able to account for the default risk premium in the implied volatility surface.

The calibration results reveal that our model is able to produce implied volatility surfaces that match the data closely.
 We compare the implied volatility surfaces that our model produces to those of \cite{ronnie-timescale}. We see that even for longer maturities our model has a prominent skew: compare  Figures~\ref{fig:imp7p} and \ref{fig:impvfouque}. Even when we ignore the stochastic volatility effects, our model fits the implied volatility of the Ford Motor Company well and performs better than the model of \cite{ronnie-timescale}; see Figure~\ref{fig:impvolFord}. This points to the importance of accounting for the default risk for companies with low ratings. 

Once the model parameters are calibrated, the model can be used to compute the prices of more exotic options.
To test whether our model produces correct prices we use the CDS spread data and show that the model implied CDS spread matches the ``out of sample" CDS data.
To compute the CDS spread, under our assumption on the recovery, one needs to reconstruct the term structure of the treasury and the corporate bonds. Moreover, one needs to separate the loss rate from the other parameters in the bond pricing formula (see \eqref{eq:cds-form} or \eqref{eq:imcdssp} for the CDS spread formula). This separation is possible since we calibrate our model to corporate bond data and stock option data jointly as described above.  The model-implied CDS spread 
time series matches the observed CDS spread time series of Ford Motor Company for over a long period of time; see Figures~\ref{fig:cds3yr} and \ref{fig:cds5yr}. This is an interesting observation since we did not make use of the CDS spread data in our calibration. This observation also shows that one can use our model to trade across different markets that contain information about the default risk of a given firm.

Our model has three building blocks: (1) We model the default event using the multi-scale stochastic intensity model of \cite{papa}. We also model the interest rate using an Ornstein-Uhlenbeck process (Vasicek model). As it was demonstrated in \cite{papa}, these modeling assumptions are effective in capturing the corporate yield curve; (2) We take the stock price process to follow a stochastic volatility model which jumps to zero when the company defaults. This stock price model was considered in \cite{bayraktar2008}. Our model specification for the stock price differs from the jump to default models for the stock price considered by \cite{carr-linetsky} and \cite{Linetsky}, which take the volatility and the default intensity to be functions of the stock price; (3) We also account for the stochastic volatility in the modeling of the stocks since even the index options (when there is no risk of default) possess implied volatility skew.
We model the volatility using the fast scale stochastic volatility model of \cite{sircar}. We demonstrate on index options (when there is no risk of default) that (see Section~\ref{sec:demnst}), we match the performance of  the two time scale volatility model \cite{ronnie-timescale}.  
The latter model extends \cite{sircar} by including a slow factor in the volatility to get a better fit to longer maturity option. We see from Section~\ref{sec:demnst} that
when one assumes the interest rate to be stochastic, the calibration performance of the stochastic volatility model with only the fast factor is as good as the two scale stochastic volatility model, which is why we choose the volatility to be driven by only the fast factor.  Even though the interest rate is stochastic in our model, we are able to obtain explicit asymptotic pricing formulas for stock options. Thanks to these explicit pricing formulas the inverse problem that we face in calibrating to the corporate bond and stock data can be solved with considerable ease.  Our modeling framework can be thought of as a hybrid of the models of \cite{sircar}, which only considers pricing options in a stochastic volatility model with constant interest rate, and \cite{papa}, which only considers a framework for pricing derivatives on bonds. Neither of these models has the means to transfer information from the equity markets to bond market or vice versa, which we are set to do in this paper. We should also note that our model also takes input from the treasury yield curve, historical stock prices, and historical spot rate data to estimate some of its parameters (see Section~\ref{sec:calibration}). 

Our model extends \cite{bayraktar2008} by taking the interest rate process to be stochastic, which leads to a richer theory and more calibration parameters, and therefore, better fit to data: (i) When the interest rate is deterministic  the corporate bond pricing formula turns out to be very crude and does not fit the bond term structure well (compare (2.57) in \cite{bayraktar2008} and \eqref{eq:apprx-bndprice}); (ii) With deterministic interest rates the bond pricing and the stock option pricing formulas share only one common term, ``the average intensity of default'' (this parameter is multiplied by the loss rate in the bond pricing equation, under our loss assumptions). Therefore, the default premium in the implied volatility surface is not accounted for as much as it should be. And our calibration analysis demonstrates that this has a significant impact.
When the volatility is taken to be constant,
both our new model and the model in \cite{bayraktar2008} have three free parameters.
The model in  \cite{bayraktar2008} produces a below par fit to the implied volatility surface (see e.g. Figure 5 in that paper), whereas our model produces an excellent fit (see Section~\ref{sec:Fordsimplied} and Figure~\ref{fig:impvolFord}); (iii) To calculate the CDS spread, in the constant interest rate model, one needs to separate the loss rate and the average intensity of default. This is again established calibrating the model to the bond term structure data and the stock option implied volatility surface. The estimates for the average intensity and the loss rate are not as accurate in \cite{bayraktar2008} as it is in our model because of (i) and (ii). This crude estimation leads to a poor out of sample match to the CDS spread time series. 
  
The other defaultable stock models are those of  \cite{carr-linetsky}, \cite{Linetsky} and \cite{carr-wu}, which assume that the interest rate is deterministic.
\cite{carr-linetsky}, \cite{Linetsky} take the volatility and the intensity to be functions of the stock price and obtain a one-dimensional diffusion for the pre-default stock price evolution. Using the fact that the resolvents of particular Markov processes can be computed explicitly, they obtain pricing formulas for stock option prices. On the other hand \cite{carr-wu} 
uses a CIR stochastic volatility model and also models the intensity to be a function of the volatility and another endogenous CIR factor. The option prices in this framework are computed numerically using inverse the Fourier transform.
We, on the other hand, use asymptotic expansions to provide explicit pricing formulas for stock options in a framework that combines a) the Vasicek interest rate model, b) fast-mean reverting stochastic volatility model, c) defaultable stock price model, d) multi-scale stochastic intensity model.

Our calibration exercise differs from that of \cite{carr-wu} since they perform a time series analysis to obtain the parameters of the underlying factors (from the the stock option prices and credit default swap spread time series), whereas we calibrate our pricing parameters to the daily implied volatility surface and bond term structure data.  Our purpose is to find a risk neutral model that matches a set of observed market prices. This risk neutral model can then be used to price more exotic, illiquid or over-the-counter derivatives. For further discussion of this calibration methodology we refer to \cite{cont} (see Chapter 13), \cite{sircar}, \cite{ronnie-timescale} and \cite{papa}.
We also provide daily prediction of the CDS spread only using the data from the bond term structure and implied volatility surface of the options. 

The rest of the paper is organized as follows: In Section 2, we introduce our modeling framework and describe the credit and equity derivatives we will consider and obtain an expression for the CDS spread under the assumption that the recovery rate of a bond that defaults is a constant fraction of its predefault value. In Section 3, we introduce the asymptotic expansion method. We obtain explicit (asymptotic) prices for bonds and equity options in Section  3.3. In Section 4, we describe the calibration of our parameters and discuss our empirical results. Figures, which show our calibration results, are located after the references.

\section{A Framework for Pricing Equity and Credit Derivatives}
\subsection{The model} 

Let $(\Omega, \H, \P)$ be a complete probability space supporting (i) correlated standard Brownian motions $\vec{W}_t=(W_t^0,W_t^1,W_t^2,W_t^3,W_t^4)$, $t \geq 0$, 
with
\begin{equation}
\E[W^0_t,W^i_t]=\rho_i t, \quad
\E[W^i_t,W^j_t]=\rho_{ij} t,\quad i,j \ \in \{1,2,3,4\}, \; t \geq 0,
\end{equation}
for some constants $\rho_i, \rho_{ij} \in (-1,1)$, 
and (ii) a Poisson process $N$ independent of $\vec{W}$. 
Let us introduce the Cox process (time-changed Poisson process) $\tilde{N}_t \triangleq N(\int_0^{t}\la_s ds)$, $t \geq 0$, where
\begin{equation}
\begin{split}
\lambda_t&=f(Y_t,Z_t), \\
dY_t&=\frac{1}{\epsilon}(m-Y_t)dt+\frac{\nu\sqrt{2}}{\sqrt{\epsilon}}dW_t^2, \quad Y_0=y,\\
dZ_t&=\delta c(Z_t)dt+\sqrt{\delta}g(Z_t)dW_t^3, \quad Z_0=z,\\
\end{split}
\end{equation}
in which $\eps, \delta$ are (small) positive constants and $f$ is a strictly positive, bounded, smooth function. We also assume that the functions $c$ and $g$ satisfy Lipschitz continuity and growth conditions so that the diffusion process for $Z_t$ has a unique strong solution.
We model the time of default as
\begin{equation}
\tau= \inf\{t \geq 0: \tilde{N}_t=1\}.
\end{equation}
We also take interest rate to be stochastic and model it as an Ornstein-Uhlenbeck process
\begin{equation}
dr_t=(\alpha-\beta r_t)dt+\eta dW_t^1, \quad r_0=r,
\end{equation}
for positive constants $\alpha$, $\beta$, and $\eta$.

We model the stock price as the solution of 
the 
stochastic differential equation
\begin{equation}
d\bar{X}_t=\bar{X}_t\left(r_t dt+\sigma_t
dW^{0}_t-d\left(\tilde{N}_t-\int_0^{t \wedge \tau}\lambda_u
du\right)\right), \quad \bar{X}_0=x,
\end{equation}
where the volatility is stochastic and is defined through
\begin{equation}
\begin{split}
\sigma _t=\sigma (\tilde{Y}_t);\quad 
d\tilde{Y}_t=\left(\frac{1}{\epsilon}(\tilde{m}-\tilde{Y}_t)-\frac{\tilde{\nu}\sqrt{2}}{\sqrt{\epsilon}}\Lambda(\tilde{Y}_t)\right)dt+\frac{\tilde{\nu}\sqrt{2}}{\sqrt{\epsilon}}dW_t^4, \quad \tilde{Y}_0=\tilde{y}.
\end{split}
\end{equation}
Here, $\Lambda$ is a smooth, bounded function of one variable which represents the market price of volatility risk. The function $\sigma$ is also a bounded, smooth function. 
Note that the discounted stock price is a martingale under the measure
$\mathbb{P}$, and at the time of default, the stock price jumps down to zero. 
 The pre-banktruptcy stock price coincides with the solution of
\begin{equation}
dX_t=(r_t+\lambda_t)X_tdt+\sigma_t X_tdW_t^0, \quad X_0=x.
\end{equation}

It will be useful to keep track of different flows of information.
Let $\mathbb{F}=\{\mathcal{F}_t, t\geq 0\}$ be the natural filtration of $\vec{W}$.
Denote the default indicator process by $I_t=1_{\{\tau \leq t\}}$, $t \geq 0$, and let $\mathbb{I}=\{\mathcal{I}_t, t \geq 0\}$ be the filtration generated by $I$. Finally, let $\mathbb{G}=\{\mathcal{G}_t, t \geq 0\}$ be an enlargement of $\mathbb{F}$ such that $\mathcal{G}_t=\mathcal{F}_t \vee \mathcal{I}_t$, $t \geq 0$.  

Since we will take $\eps$ and $\delta$ to be small positive constants, the processes $Y$ and $\tilde{Y}$ are fast mean reverting, and $Z$ evolves on a slower time scale. See
\cite{ronnie-timescale} for an exposition and motivation of multi-scale modeling
in the context of stochastic volatility models.

We note that our specification of the intensity of default coincides with that of \cite{papa}, who considered only a framework for pricing credit derivatives. Our stock price specification is similar to that of \cite{Linetsky} and \cite{carr-linetsky} who considered a framework for only pricing equity options on defaultable stocks. Our volatility specification, on the other hand, is in the spirit of \cite{sircar}.

 \cite{bayraktar2008} considered a similar modeling framework to the one considered here, but the interest rate was taken to be deterministic. In this paper, by extending this modeling framework to incorporate stochastic interest rates, we are able to consistently price credit and equity derivatives and produce more realistic yield curve and implied volatility surfaces. We are also able to take the equity option surface and the yield curve data as given and predict the credit default swap spread on a given day. Testing our model prediction against real data demonstrates the power of our pricing framework.

\subsection{Equity and credit derivatives}
In our framework, we will price European options, bonds, and credit default swaps of the same company in a consistent way.

\textbf{1.} The price of a  European call option with maturity $T$ and strike price $K$ is given by
\begin{equation}
\begin{split}
C(t;T,K)&=\mathbb{E}\left[\exp\left(-\int_t^T r_s
ds\right)(\bar{X}_T-K)^+1_{\{\tau>T\}}\bigg| \G_t \right]
\\&=1_{\{\tau>t\}} \mathbb{E}\left[\exp\left(-\int_t^T (r_s+\lambda_s)
ds\right)(X_T-K)^+\bigg| \F_t\right],
\end{split}
\end{equation}
in which the equality follows from Lemma 5.1.2 of \cite{bielecki-rutkowski}. (This lemma, which lets us write a conditional expectation with respect to $\G_t$ in terms of conditional expectations with respect to $\F_t$, will be used in developing several identities below). Also, see \cite{Linetsky} and \cite{carr-linetsky} for a similar computation. 

On the other hand, the price of a put option with the same maturity and strike price is
{\small
\begin{equation}
\begin{split}
\text{Put}(t;T)&=\mathbb{E}\left[\exp\left(-\int_t^T r_s
ds\right)(K-X_T)^+1_{\{\tau>T\}}\bigg| \G_t\right]+\mathbb{E}\left[\exp\left(-\int_t^T
r_sds\right)K1_{\{\tau\leq T\}}\big| \G_t\right]\\
&=1_{\{\tau>t\}}\Bigg(\mathbb{E}\left[\exp\left(-\int_t^T (r_s+\lambda_s)
ds\right)(K-X_T)^+\bigg| \F_t\right]\\&+K\mathbb{E}\left[\exp\left(-\int_t^T
r_sds\right)\bigg|\F_t\right]-K\mathbb{E}\left[\exp\left(-\int_t^T
(r_s+\lambda_s)ds\right)\bigg|\F_t\right]\Bigg).
\end{split}
\end{equation}}
\textbf{2.} Consider a defaultable bond with maturity $T$ and par value of 1 dollar. We assume the recovery of the market value, introduced by \cite{duffie-singleton1999}. In this model, if the issuer company defaults prior to maturity, the holder of the bond recovers a constant fraction $1-l$ of the pre-default value, with $l \in [0,1]$. The price of such a bond is
\begin{equation}\label{eq:corporate-bond}
\begin{split}
B^{c}(t;T)&=\E \Bigg[\exp\left(-\int_t^{T}r_s ds\right)1_{\{\tau>T\}}
+\exp\left(-\int_t^{\tau}r_s ds\right)1_{\{\tau \leq T\}} \, (1-l) B^c(\tau-;T)\bigg| \G_t\Bigg] \\
&=\E\left[\exp\left(-\int_t^T (r_s
+l\,\lambda_s)ds\right)\bigg| \F_t \right],
\end{split}
\end{equation}
on $\{\tau>t\}$, see \cite{duffie-singleton1999} and \cite{schonbucher1998}.

\textbf{3.} Consider a credit default swap (CDS) written on $B^c$, which is a insurance against losses incurred upon default from holding a corporate bond.
The protection buyer pays a fixed premium, the so-called CDS spread, to the protection seller. The premium is paid on fixed dates 
$\mathcal{T}=(T_1,\cdots,T_M)$, with
$T_M$ being the maturity of the CDS contract. We denote the
CDS spread at time $t$ by $c^{ds}(t;\mathcal{T})$. Our purpose is to determine a fair value for the CDS spread so that what the protection buyer buyer is expected to pay, the value of the premium leg of the contract, is equal to what the protection seller is expected to pay, the value of the protection leg of the contract. For a more detailed description of the CDS contract, see \cite{bielecki-rutkowski} or \cite{schoenbucher}.

The present value of the premium leg of the contract is
\begin{equation}\label{eq:prem-leg}
\begin{split}
\text{Premium}(t;\mathcal{T})&=c^{ds}(t;\mathcal{T})\,\E\left[\displaystyle\sum_{m=1}^M\exp\left(-\int_t^{T_m}
r_s ds\right)1_{\{\tau>T_m\}}\bigg|\G_t\right]\\
&=1_{\{\tau>t\}}c^{ds}(t;\mathcal{T})\displaystyle\sum_{m=1}^M\E\left[\exp\left(-\int_t^{T_m}
(r_s+\lambda_s) ds\right)\bigg| \F_t\right],
\end{split}
\end{equation}
in which we assumed that $t<T_1$.
The present value of the protection leg of the contract under our
assumption of {\it recovery of market value} is
\begin{equation}\label{eq:protection-leg}
\text{Protection}(t;\mathcal{T})=1_{\{\tau>t\}}\E\left[\exp\left(-\int_t^\tau
r_sds\right)1_{\{\tau\leq T_M\}}l\,B^c(\tau-;T_M)\bigg| \G_t\right]
\end{equation}

Adding (\ref{eq:corporate-bond}) and (\ref{eq:protection-leg}), we obtain
\begin{equation}\label{eq:prt-leg+plus-bond}
\begin{split}
\text{Protection}(t;\mathcal{T})+B^c(t;T_M)&=\E \Bigg[\exp\left(-\int_t^{T}r_s ds\right)1_{\{\tau>T\}}
+\exp\left(-\int_t^{\tau}r_s ds\right)1_{\{\tau \leq T\}}  B^c(\tau-;T)\bigg| \G_t\Bigg]\\
&=1_{\{\tau>t\}}\E\left[\exp\left(-\int_t^{T_M}
r_sds\right)\bigg| \F_t\right],
\end{split}
\end{equation}
where the last equality is obtained by setting $l=0$ in (\ref{eq:corporate-bond}).

Now, the CDS spread can be determined, by setting
$\text{Protection}(t;\mathcal{T})=\text{Premium}(t;\mathcal{T})$ and using equations (\ref{eq:prem-leg}) and (\ref{eq:prt-leg+plus-bond}),
as
\begin{equation}\label{eq:cds-form}
c^{ds}(t;\mathcal{T})=1_{\{\tau>t\}}\frac{\mathbb{E}\left[\exp\left(-\int_t^{T_M}
r_sds\right)\bigg|\F_t\right]-\mathbb{E}\left[\exp\left(-\int_t^{T_M}
(r_s
+l\,\lambda_s)ds\right)\bigg|\F_t\right]}{\displaystyle\sum_{m=1}^M\E\left[\exp\left(-\int_t^{T_m}
(r_s+\lambda_s) ds\right)\bigg|\F_t\right]}.
\end{equation}

\section{Explicit Pricing Formulas for Credit and Equity Derivatives}
\subsection{Pricing equation}
Let $P^{\epsilon,\delta}$ denote 
\begin{equation}\label{eq:Peps-delta}
P^{\epsilon,\delta}(t,X_t,r_t,Y_t, \tilde{Y}_t,Z_t)=\E\left[\exp\left(-\int_t^T
(r_s+l \lambda_s) ds\right)h(X_T)\bigg|\F_t\right].
\end{equation}
When $l=1$ and $h(X_T)=(X_T-K)^+$, $P^{\eps,\delta}$ is the price of a call option (on a defaultable stock). On the other hand, when $h(X_T)=1$, $P^{\eps,\delta}$ becomes the price of a defaultable bond.

Using the Feynman-Kac formula, we can characterize $P^{\eps,\delta}$ as the solution of
\begin{equation}\label{mainEq}
\begin{aligned}
&\L^{\epsilon,\delta}P^{\epsilon,\delta}(t,x,r,y,\tilde{y},z)=0, \\
&P^{\epsilon,\delta}(T,x,r,y,\tilde{y},z)=h(x),
\end{aligned}
\end{equation}
where the partial differential operator $\L^{\eps,\delta}$ is defined as 
\begin{gather}\label{L}
\L^{\epsilon,\delta} \triangleq \frac{1}{\epsilon}\L_0+\frac{1}{\sqrt{\epsilon}}\L_1+\L_2+\sqrt{\delta}\M_1+\delta
\M_2+\sqrt{\frac{\delta}{\epsilon}}\M_3,
\end{gather}
in which
\begin{equation*}
\begin{split}
\L_0& \triangleq \nu^2\frac{\partial^2}{\partial y^2}+(m-y)\frac{\partial}{\partial y}+\tilde{\nu}^2\frac{\partial^2}{\partial \tilde{y}^2}+(\tilde{m}-\tilde{y})\frac{\partial}{\partial \tilde{y}}+2\rho_{24} v\tilde{v}\frac{\partial^2}{\partial y\partial \tilde{y}},\\
\L_1& \triangleq \rho_2\sigma(\tilde{y})\nu\sqrt{2}x\frac{\partial^2}{\partial x \partial y}+\rho_{12}\eta\nu\sqrt{2}\frac{\partial^2}{\partial r \partial y}+\rho_4\sigma(\tilde{y})\tilde{\nu}\sqrt{2}x\frac{\partial^2}{\partial x \partial \tilde{y}}+\rho_{14}\eta\tilde{\nu}\sqrt{2}\frac{\partial^2}{\partial r \partial \tilde{y}}-\Lambda(\tilde{y})\tilde{\nu}\sqrt{2}\frac{\partial}{\partial \tilde{y}},\\
\L_2 &\triangleq \frac{\partial}{\partial
t}+\frac{1}{2}\sigma^2(\tilde{y})x^2\frac{\partial^2}{\partial
x^2}+(r+f(y,z))x\frac{\partial}{\partial x}+(\alpha-\beta
r)\frac{\partial}{\partial r}+\sigma(\tilde{y})\eta\rho_1x\frac{\partial^2}{\partial x\partial r}+\frac{1}{2}\eta^2\frac{\partial^2}{\partial r^2}-(r+l\,f(y,z))\cdot,
\end{split}
\end{equation*}
\begin{equation*}
\begin{split}
\M_1& \triangleq \sigma(\tilde{y})\rho_3g(z)x\frac{\partial^2}{\partial x\partial z}+\eta\rho_{13}g(z)\frac{\partial^2}{\partial r\partial z},\quad
\M_2 \triangleq c(z)\frac{\partial}{\partial z}+\frac{1}{2}g^2(z)\frac{\partial^2}{\partial z^2},\\
\M_3& \triangleq \rho_{23}\nu\sqrt{2}g(z)\frac{\partial^2}{\partial y\partial
z}+\rho_{34}\tilde{\nu}\sqrt{2}g(z)\frac{\partial^2}{\partial
\tilde{y}\partial z}.
\end{split}
\end{equation*}
\subsection{Asymptotic expansion}
We construct an asymptotic expansion for $P^{\eps,\delta}$ as $\eps,\delta \rightarrow 0$.
First, we consider an expansion of $P^{\epsilon,\delta}$ in powers
of $\sqrt{\delta}$
\begin{equation}\label{eq:delta-expansion}
P^{\epsilon,\delta}=P_0^\epsilon+\sqrt{\delta}P_1^\epsilon+\delta
P_2^\epsilon+\cdots
\end{equation}
By inserting (\ref{eq:delta-expansion}) into (\ref{mainEq}) and comparing the $\delta^0$ and $\delta$ terms, we obtain that $P_0^\epsilon$ satisfies
\begin{align}\label{slow1}
&\left(\frac{1}{\epsilon}\L_0+\frac{1}{\sqrt{\epsilon}}\L_1+\L_2\right)P_0^\epsilon=0,\\
&P_0^\epsilon(T,x,r,y, \tilde{y},z)=h(x),\notag
\end{align}
and that $P_1^{\eps}$ satisfies
\begin{align}\label{slow2}
&\left(\frac{1}{\epsilon}\L_0+\frac{1}{\sqrt{\epsilon}}\L_1+\L_2\right)P_1^\epsilon=-\left(\M_1+\frac{1}{\sqrt{\epsilon}}\M_3\right)P_0^\epsilon,\\
&P_1^\epsilon(T,x,y,\tilde{y},z,r)=0.\notag
\end{align}

Next, we expand the solutions of (\ref{slow1}) and (\ref{slow2}) in powers of $\sqrt\epsilon$
\begin{align}
&P_0^\epsilon=P_0+\sqrt{\epsilon}P_{1,0}+\epsilon
P_{2,0}+\epsilon^{3/2}P_{3,0}+\cdots\\
&P_1^\epsilon=P_{0,1}+\sqrt{\epsilon}P_{1,1}+\epsilon
P_{2,1}+\epsilon^{3/2}P_{3,1}+\cdots \label{eq:P1eps-expansion}
\end{align}
Inserting the expansion for $P_0^\epsilon$ into (\ref{slow1}) and
matching the $1/\eps$ terms gives $\L_0 P_0=0$. We choose $P_0$ not to depend on $y$ and $\tilde{y}$ because the other solutions have exponential growth at infinity (see e.g. \cite{ronnie-timescale}).
Similarly, by matching the $1/\sqrt{\eps}$ terms in (\ref{slow1}) we obtain that
$\L_0P_{1,0}+\L_1P_0=0$. Since $\L_1$ takes derivatives only with respect to $y$ and $\tilde{y}$, we observe that $\L_0P_{1,0}=0$. We choose $P_{1,0}$ not to depend on $y$ and $\tilde{y}$.

Now equating the order-one terms in the expansion of (\ref{slow1}) and using the
fact that $\L_1P_{1,0}=0$, we get that
\begin{equation}\label{eq:eqorderone}
\L_0P_{2,0}+\L_2P_0=0,
\end{equation}
which is a Poisson equation for $P_{2,0}$ (see e.g. \cite{sircar}). The solvability
condition for this equation requires that
\begin{equation}\label{eq:PDE-P0}
\langle \L_2\rangle P_0=0,
\end{equation}
where $\langle \cdot \rangle$ denotes the averaging with
respect to the invariant distribution of $(Y_t,\tilde{Y}_t)$,
whose density is given by
\begin{equation}
\Psi(y,\tilde{y})=\frac{1}{2\pi\nu\tilde{\nu}}\exp\left\{-\frac{1}{2(1-\rho_{24}^2)}\left[\left(\frac{y-m}{\nu}\right)^2+\left(\frac{\tilde{y}-\tilde{m}}{\tilde{\nu}}\right)^2-2\rho_{24}\frac{(y-m)(\tilde{y}-\tilde{m})}{\nu\tilde{\nu}}\right]\right\}.
\end{equation}
Let us denote
\begin{equation}
\bar{\sigma}_1 \triangleq \langle \sigma(\tilde{y}) \rangle, \quad 
\bar{\sigma}_2^2 \triangleq \langle \sigma^2(\tilde{y}) \rangle, \quad \bar{\lambda}(z)=\langle f(y,z)\rangle.
\end{equation}
To demonstrate the effect of averaging on $\L_2$, let us write
\begin{equation}
\langle \L_2 \rangle:=\frac{\partial}{\partial
t}+\frac{1}{2}\bar{\sigma}_2^2x^2\frac{\partial^2}{\partial
x^2}+(r+\bar{\lambda}(z))x\frac{\partial}{\partial
x}+(\alpha-\beta
r)\frac{\partial}{\partial r}
+\bar{\sigma}_1\eta\rho_1x\frac{\partial^2}{\partial x\partial
r}+\frac{1}{2}\eta^2\frac{\partial^2}{\partial
r^2}-(r+l\,\bar{\lambda}(z))\cdot
\end{equation}
Together with the terminal condition
\begin{equation}\label{eq:bndry-cond-P0}
P_0(T,x,r,z)=h(x),
\end{equation}
equation (\ref{eq:PDE-P0}) defines the leading order term $P_0$.
On the other hand from (\ref{eq:eqorderone}), we can also deduce that
\begin{equation}\label{eq:p20}
P_{2,0}=-\L_0^{-1}(\L_2-\langle \L_2\rangle)P_0.
\end{equation}
Matching the $\sqrt{\eps}$ order terms in the expansion of (\ref{slow1}) yields
\begin{equation}
\L_0P_{3,0}+\L_1P_{2,0}+\L_2P_{1,0}=0,
\end{equation}
which is a Poisson equation for $P_{3,0}$. The solvability
condition for this equation requires that
\begin{equation}\label{eq:PDE-P20}
\langle \L_2P_{1,0}\rangle=-\langle \L_1P_{2,0}\rangle=\langle \L_1\L_0^{-1}(\L_2-\langle \L_2\rangle)\rangle P_0,
\end{equation}
which along with the terminal condition 
\begin{equation}\label{eq:PDE-P20-b}
P_{1,0}(T,x,r,z)=0,
\end{equation}
completely identifies the function $P_{1,0}$.
To obtain the second equality in (\ref{eq:PDE-P20}) we used (\ref{eq:p20}).

Next, we will express the right-hand side of (\ref{eq:PDE-P20}) more explicitly. To this end, 
 let $\psi$, $\kappa$, and $\phi$ be the solutions of the Poisson equations
\begin{equation}\label{eq:Poisson-equations}
\L_0\psi(\tilde{y})=\sigma(\tilde{y})-\bar{\sigma}_1 \quad
\L_0\kappa(\tilde{y})=\sigma^2(\tilde{y})-\bar{\sigma}_2^2, \quad \text{and} \quad \L_0\phi(y,z)=(f(y,z)-\bar{\lambda}(z)),
\end{equation}
respectively.
First observe that
\begin{equation}
(\L_2-\langle \L_2\
\rangle) P_0=\frac{1}{2}(\sigma^2(\tilde{y})-\bar{\sigma}_2^2)x^2\frac{\partial^2P_0}{\partial
x^2}+(\sigma(\tilde{y})-\bar{\sigma}_1)\eta\rho_1x\frac{\partial
^2P_0}{\partial x\partial r}
+l\,(f(y,z)-\bar{\lambda}(z))\left(x\frac{\partial P_0}{\partial x }-P_0\right).
\end{equation}
Now, along with (\ref{eq:Poisson-equations}), we can write
\begin{equation}
\L_0^{-1}(\L_2-\langle \L_2\rangle)P_0=\frac{1}{2}\kappa(\tilde{y})
x^2\frac{\partial^2 P_0}{\partial
x^2}+\psi(\tilde{y})\eta\rho_1x\frac{\partial ^2 P_0}{\partial
x\partial r} +l\,\phi(y,z)\left(x\frac{\partial P_0}{\partial
x}-P_0\right).
\end{equation}
Applying the differential operator $\L_1$ to the last expression yields 
\begin{equation}\label{p10rhs}
\begin{split}
\langle \L_1 \L_0^{-1}(\L_2-\langle
\L_2\rangle)\rangle P_0
&=l\,\rho_2\nu\sqrt{2}\langle\sigma\phi_y\rangle
x^2\frac{\partial P_0}{\partial
x^2}+l\,\rho_{12}\eta\nu\sqrt{2}\langle\phi_y\rangle\frac{\partial} {\partial
r} \left(x\frac{\partial P_0}{\partial
x}-P_0\right)\\
&+\rho_4\tilde{\nu}\sqrt{2}\left(\frac{1}{2}\langle\sigma\kappa_{\tilde{y}}\rangle
x\frac{\partial}{\partial x}\left(x^2\frac{\partial^2 P_0}{\partial
x^2}\right)+\langle\sigma\psi_{\tilde{y}}\rangle\eta\rho_1x\frac{\partial} {\partial
x}\left(x\frac{\partial^2 P_0}{\partial x\partial
r}\right)\right)\\
&+\rho_{14}\eta\tilde{\nu}\sqrt{2}\left(\frac{1}{2}\langle\kappa_{\tilde{y}}\rangle\frac{\partial}{\partial
r}\left(x^2\frac{\partial^2 P_0}{\partial
x^2}\right)+\langle\psi_{\tilde{y}}\rangle\eta\rho_1\frac{\partial}{\partial
r}\left(x\frac{\partial^2 P_0}{\partial x\partial
r}\right)\right)\\
&-\tilde{\nu}\sqrt{2}\left(\frac{1}{2}\langle\Lambda\kappa_{\tilde{y}}\rangle
x^2\frac{\partial P_0}{\partial
x^2}+\langle\Lambda\psi_{\tilde{y}}\rangle\eta\rho_1x\frac{\partial^2 P_0}{\partial
x\partial r}\right).
\end{split}
\end{equation}

Finally, we insert the expression for $P_1^{\eps}$ in (\ref{eq:P1eps-expansion}) into (\ref{slow2}) and collect the terms with the same powers of $\eps$. Arguing as before, we obtain that
$P_{0,1}$ is independent of
$y$ and $\tilde{y}$ and satisfies:
\begin{equation}\label{p2Eq}
\begin{split}
\langle \L_2\rangle P_{0,1}&=-\langle \M_1\rangle P_0,\quad 
P_{0,1}(T,x)=0.
\end{split}
\end{equation}

\subsection{Explicit pricing formula}\label{sec:explct}

We approximate $P^{\eps,\delta}$ defined in (\ref{eq:Peps-delta}) by
\begin{equation}\label{eq:apprx}
\widetilde{P}^{\eps, \delta}=P_0+ \sqrt{\eps}P_{1,0}+\sqrt{\delta}P_{0,1}.
\end{equation}
Since the Vasicek interest rate process is unbounded, which implies that the potential term in $\L_2$ or the discounting term in \eqref{eq:Peps-delta} is unbounded, the arguments of 
 \cite{ronnie-timescale} can not be directly used. However as in \cite{MR2046929} and \cite{papa}, one can write
 \begin{equation}
 P^{\epsilon,\delta}(t,X_t,r_t,Y_t, \tilde{Y}_t,Z_t)=B(t,T)\E^T\left[\exp\left(-\int_t^T
l \lambda_s ds\right)h(X_T)\bigg|\F_t\right]:=B(t,T)F^{\epsilon,\delta}(t,X_t,r_t,Y_t, \tilde{Y}_t,Z_t),
 \end{equation}
in which 
\begin{equation}\label{eq:RN}
\frac{d\P^{T}}{d\P}=\frac{\exp\left(-\int_0^{T}r_s ds\right)}{B(0,T)},
\end{equation}
and
\begin{equation}
B(t,T)=\E\left[\exp\left(-\int_t^{T}r_s ds\right)\bigg|\F_t\right].
\end{equation}
 Now, the analysis of  \cite{ronnie-timescale} can be used to approximate $F^{\epsilon,\delta}(t,x,r,y,\tilde{y},z)$.
 As a result of this analysis for each $(t,x,r,y,\tilde{y},z)$, there exists a constant $C$ such that
$|P^{\eps,\delta}-\widetilde{P}^{\eps,\delta}| \leq C \cdot (\eps+\delta)$ when $h$ is smooth, and
$|P^{\eps,\delta}-\widetilde{P}^{\eps,\delta}| \leq C \cdot (\eps \log(\eps)+\delta+ \sqrt{\eps \delta})$ when $h$ is a put or a call pay-off. In what follows, we will obtain $P_0$, $P_{1,0}$ and $P_{0,1}$ explicitly.

Our first objective is to develop a closed-form expression for $P_0$, the solution of (\ref{eq:PDE-P0}) and (\ref{eq:bndry-cond-P0}). 
\begin{prop}
The leading order term $P_0$ in (\ref{eq:apprx}) is given by:
\begin{align}\label{gp}
P_0(t,x,z,r)&=B
^c_0(t,r;z,T,l)\int_{-\infty}^{\infty}h(\exp(u))\frac{1}{\sqrt{2\pi
v(t,T)}}\exp(-\frac{(u-m(t,T))^2}{2v(t,T)})du,
\end{align}
where
\begin{equation}\label{btt} B
^c_0(t,r;z,T,l) \triangleq \exp\big(-l\bar{\lambda}(z)(T-t)+ a(T-t)-b(T-t)r),
\end{equation}
in which the functions $a(s)$ and $b(s)$ are defined as:
\begin{equation}
\begin{split}
a(s)=\left(\frac{\eta^2}{2\beta^2}-\frac{\alpha}{\beta}\right)s+\left(\frac{\eta^2}{\beta^3}-\frac{\alpha}{\beta^2}\right)(\exp(-\beta
s)-1)-\frac{\eta^2}{4\beta^3}(\exp(-2\beta s)-1)
\end{split}
\end{equation}
and $b(s)=(1-\exp(-\beta s))/\beta$. 
On the other hand,
\begin{equation}
\begin{split}
v_{t,T}&=\left(\bar{\sigma}_2^2+\frac{2\eta \rho_1\bar\sigma_1}{\beta}+\frac{\eta^2}{\beta^2}\right)(T-t)+\left(\frac{2\eta\rho_1\bar\sigma_1}{\beta^2}+\frac{2\eta^2}{\beta^3}\right)\exp(-\beta(T-t))\\
&-\frac{\eta^2}{2\beta^3}\exp(-2\beta(T-t))-\left(\frac{2\eta\rho_1\bar\sigma_1}{\beta^2}+\frac{3\eta^2}{2\beta^3}\right),
\end{split}
\end{equation}
and 
\begin{equation}
m_{t,T}=\log(x)+\bar{\lambda}\cdot(T-t)-a(T-t)+b(T-t)r-\frac{1}{2}v(t,T).
\end{equation}
\end{prop}
\begin{proof}
By applying the Feynman-Kac theorem to (\ref{eq:PDE-P0}) and (\ref{eq:bndry-cond-P0}) we have that 
\begin{equation}
\begin{split}
P_0(t,x,z,r)&= \mathbb{E}\left[\exp\left(-\int_t^T (r_s+l\bar{\lambda}(z))
ds\right)h(S_T)\bigg| S_t=x, r_t=r\right],
\end{split}
\end{equation}
where the dynamics of $S$ is given by
\begin{equation}\label{eq:effective}
\begin{split}
dS_t=(r_t+\bar{\lambda}(z))S_tdt+\bar{\sigma}_2 S_td\widetilde{W}_t^0, \quad 
\end{split}
\end{equation}
in which $\widetilde{W}^0$ is a Wiener process whose correlation with $W^1$ is
$\bar{\rho}_1=\frac{\bar{\sigma}_1}{\bar{\sigma}_2}\rho_1$.

Let us define
\begin{equation}
\tilde{P}_0(t,x,r)=\E\left[\exp\left(-\int_t^T r_s
ds\right)h(\widetilde{S}_T)\bigg| \widetilde{S}_t=x, r_t=r\right],
\end{equation}
in which
\begin{equation}
d\widetilde{S}_t=r_t \widetilde{S}_tdt+\bar{\sigma}_2 \widetilde{S}_td\widetilde{W}_t^0. \quad 
\end{equation}
Then 
\begin{equation}
P_0(t,x,z,r)=e^{-l\bar{\lambda}(z)\,(T-t)}\tilde{P}_0(t,x \exp(\bar{\lambda}(z) (T-t)),r).
\end{equation}

 Now, by following \cite{geman}, we change the probability measure $\P$ to the forward measure $\P^T$ through the Radon-Nikodym derivative \eqref{eq:RN}

We can obtain the following representation of $\tilde{P_0}$ using the $T$-forward measure
\begin{equation}
\tilde{P}_0(t,\widetilde{S}_t,r_t)=B(t,T)\E^T \left[h(\widetilde{S}_T)|\F_t\right]=B(t,T)\E^T\left[h(F_T)|\F_t\right],
\end{equation}
in which 
\begin{equation}\label{eq:F}
F_t \triangleq \frac{\widetilde{S}_t}{B(t,T)},
\end{equation}
which is a $\P^T$ martingale.
Note that an explicit expression for $B(t,T)$ is available since $r_t$ is a Vasicek model, and it is given in terms of the functions $a$ and $b$:
\begin{equation}
B(t,T)=\exp(a(T-t)-b(T-t)r_t).
\end{equation}
By  applying It\^{o}'s formula to (\ref{eq:F}), we observe that the dynamics of $F$ are
\begin{equation}
dF_t=F_t(\bar{\sigma}_1 d\widetilde{W}^0_t+b(T-t)\eta d\widetilde{W}_t^1),
\end{equation}
in which $\widetilde{W}^1$ is a $\P^{T}$ Brownian motion whose correlation with the $\widetilde{W}^0$ (which is still a Brownian motion under $\P^T$) is $\bar{\rho}_1$.
Given $X_t$ and $B(t,T)$, the random variable
$\log F_T$ is normally distributed with
variance
\begin{equation}
\begin{split}
v_{t,T}&=\bar{\sigma}_2^2(T-t)+\eta^2\int_t^T
b^2(T-s)ds+2\eta\bar \rho_1 \bar{\sigma}_2 \int_t^T b(T-s)ds\\
&=\left(\bar{\sigma}_2^2+\frac{2\eta\bar\rho_1\bar\sigma_2}{\beta}+\frac{\eta^2}{\beta^2}\right)(T-t)+\left(\frac{2\eta\bar\rho_1\bar\sigma_2}{\beta^2}+\frac{2\eta^2}{\beta^3}\right)\exp(-\beta(T-t))\\
&-\frac{\eta^2}{2\beta^3}\exp(-2\beta(T-t))-\left(\frac{2\eta\bar\rho_1\bar\sigma_2}{\beta^2}+\frac{3\eta^2}{2\beta^3}\right),
\end{split}
\end{equation}
and mean
\begin{equation}
m(t,T)=\log
F_t-\frac{1}{2}\int_t^T(\bar{\sigma}_2^2+b^2(T-s)\eta^2+\bar{\rho}_1\bar{\sigma}_2
b(T-s)\eta)ds
=\log \left(\frac{\widetilde{S}_t}{B(t,T)}\right)-\frac{1}{2}v_{t,T}.
\end{equation}
Now the result immediately follows.
\end{proof}

An immediate corollary of the last proposition is the following:
\begin{cor}\label{cor:p0}
i) When $l=1$, $h(x)=(x-K)^+$, then (\ref{gp}) becomes
\begin{equation}\label{eq:call-opt}
C_0(t,x,z,r)=xN(d_1)-KB_0
^c(t,r;z,T,1)N(d_2),
\end{equation}
in which $N$ is the standard normal cumulative distribution
function and
\begin{equation}
d_{1,2}=\frac{\log\frac{x}{K B
^c_0(t,r;z,T,1)} \pm \frac{1}{2}v_{t,T}}{\sqrt{v_{t,T}}}.
\end{equation}

ii) When $l=1$, and $h(x)=(K-x)^+$, then  (\ref{gp}) becomes
\begin{align}\label{eq:put-opt}
\text{Put}_0(t,x,z,r)=-x+xN(d_1)-K B_0
^c(t,r;z,T,1)N(d_2)
+K B_0^c(t,r;z,T,0).
\end{align}

iii) When $h(x)=1$, then (\ref{gp}) coincides with (3.30) in \cite{papa}.
\end{cor}

\begin{prop}\label{prop:sec-term}
The correction term $\sqrt{\epsilon}P_{1,0}$ is given by
\begin{equation}\label{p1}
\begin{split}
\sqrt{\epsilon}P_{1,0}&=-(T-t)\left(V_1^\epsilon
x^2\frac{\partial^2P_0}{\partial x^2}+V_2^\epsilon
x\frac{\partial}{\partial x}\left(x^2\frac{\partial^2P_0}{\partial
x^2}\right)\right)
\\&+l \, V_3^\epsilon\left(-x\frac{\partial^2 P_0}{\partial x \partial \alpha}-\frac{\partial P_0}{\partial \alpha}\right)+V^{\epsilon}_4x^2\frac{\partial^3P_0}{\partial
x^2 \partial \alpha}+V_5^\epsilon x\frac{\partial^2
P_0}{\partial \eta \partial x}+V_6^\epsilon x \frac{\partial^2 P_0}{\partial x \partial \alpha},
\end{split}
\end{equation}
in which
\begin{equation}
\begin{split}
V_1^\epsilon&=\sqrt{\epsilon}\left(l\,\rho_2\nu\sqrt{2}\langle \sigma\phi_y \rangle-\tilde{\nu}\sqrt{2}\frac{1}{2}\langle\Lambda\kappa_{\tilde{y}}\rangle\right), \quad
V_2^\epsilon=\frac{1}{2}\sqrt{\epsilon}\rho_4\tilde{\nu}\sqrt{2}\langle \sigma\kappa_{\tilde{y}}\rangle,\\
V_3^\epsilon&=\sqrt{\epsilon}(\rho_{12}\eta\nu\sqrt{2}\langle\phi_y\rangle),\quad V_4^\epsilon=-\sqrt{\epsilon}\left(\frac{1}{2}\rho_{14}\eta\tilde{\nu}\sqrt{2}\langle\kappa_{\tilde{y}}\rangle-\rho_4\tilde{\nu}\sqrt{2}\langle\sigma\psi_{\tilde{y}}\rangle\eta\rho_1+\rho_{14}\eta\tilde{\nu}\sqrt{2}\langle\psi_{\tilde{y}}\rangle\bar{\sigma}_1\rho_1^2\right),\\
V_5^\epsilon&=-\sqrt{\epsilon}(\rho_{14}\eta\tilde{\nu}\sqrt{2}\langle \psi_{\tilde y}\rangle\rho_1),\quad V_6^\epsilon=\sqrt{\epsilon}(-\rho_4\tilde{\nu}\sqrt{2}\langle\sigma\psi_{\tilde{y}}\rangle\eta\rho_1+\rho_{14}\eta\tilde{\nu}\sqrt{2}\langle\psi_{\tilde{y}}\rangle\bar{\sigma}_1\rho_1^2-\tilde{\nu}\sqrt{2}\langle\Lambda\psi_{\tilde{y}}\rangle\eta\rho_1).
\end{split}
\end{equation}
\end{prop}
\begin{proof}
Recall that $P_{1,0}$ is the solution of (\ref{eq:PDE-P20}) and (\ref{eq:PDE-P20-b}) and that the right-hand-side of  (\ref{eq:PDE-P20}) is given by (\ref{p10rhs}).  
The result is a simple algebraic exercise given the following four observations:

\noindent 1)  $x^n\frac{\partial^n}{\partial x^n}$ commutes
with $\langle \L_2\rangle$. 

\noindent 2) 
$-(T-t)(x^n\frac{\partial^n}{\partial x^n})P_0$ solves:
\begin{equation}
\begin{split}
\langle \L_2 \rangle u&=\left(x^n\frac{\partial^n}{\partial x^n}\right)P_0,\quad
u(T,x,r;z)=0.
\end{split}
\end{equation}
3) By differentiating (\ref{eq:bndry-cond-P0}) with respect to $\alpha$, we see
that $-\frac{\partial P_0}{\partial \alpha}$ also solves 
\begin{equation}
\begin{split}
\langle \L_2 \rangle u&=\frac{\partial P_0}{\partial r},\quad
u(T,x,r;z)=0.
\end{split}
\end{equation}
4) Using 1) and 2) above and the equation we obtain by differentiating (\ref{eq:PDE-P0}) with respect to
$\eta$, we can show that
$1/\eta \cdot (\bar{\sigma}_1\rho_1 x\frac{\partial^2
P_0}{\partial x\partial \alpha}-\frac{\partial P_0}{\partial
\eta})$ solves
\begin{equation}
\begin{split}
\langle \L_2 \rangle u&=\frac{\partial^2 P_0}{\partial r^2},\quad
u(T,x,r;z)=0.
\end{split}
\end{equation}
\end{proof}
\begin{remark}
By differentiating (\ref{eq:PDE-P0}) with respect to $r$, we obtain
\begin{equation}
\langle \L_2\rangle \frac{\partial P_0}{\partial
r}=-x\frac{\partial}{\partial x}P_0+\beta\frac{\partial P_0
}{\partial r}+P_0.
\end{equation}
Using observation 2 in the proof of Proposition~\ref{prop:sec-term},
we see that $\frac{1}{\beta}\left(-(T-t)(x\frac{\partial
P_0}{\partial x}-P_0)+\frac{\partial P_0}{\partial r}\right)$
solves
\begin{equation}
\langle \L_2 \rangle u=\frac{\partial P_0}{\partial r}, \quad
u(T,x,r;z)=0.
\end{equation}
Now, it follows from observation 3 in the proof of Proposition~\ref{prop:sec-term} that
\begin{equation}
-\frac{\partial P_0}{\partial
\alpha}=\frac{1}{\beta}\left(-(T-t)\left(x\frac{\partial P_0}{\partial
x}-P_0\right)+\frac{\partial P_0}{\partial r}\right).
\end{equation}
Using this identity, we can express (\ref{p1}) only in terms of the ``Greeks".
\end{remark}

Next, we obtain an explicit expression for $P_{0,1}$, the solution of
 (\ref{p2Eq}). We need some preparation first. 
By differentiating (\ref{eq:PDE-P0}) with respect to $z$, we see that $\frac{\partial
P_0}{\partial z}$ solves
\begin{equation}
\langle \L_2\rangle u=-\bar{\lambda}'(z)x\frac{\partial
P_0}{\partial x}+l\,\bar{\lambda}'(z)P_0, \quad 
u(T,x,r;z)=0.
\end{equation}
As a result (see Observation 2 in the proof of Propostion~\ref{prop:sec-term})
\begin{equation}
\frac{\partial P_0}{\partial
z}=(T-t)\bar{\lambda}'(z)\left(x\frac{\partial P_0}{\partial x}-l\,P_0\right),
\end{equation}
from which it follows that $-\langle \M_1 \rangle P_0$ can be represented as
\begin{equation}
-\langle
\M_1\rangle P_0=-(T-t)\bar{\lambda}'(z)\left(\bar{\sigma}_1\rho_3g(z)\left(x^2\frac{\partial^2
P_0}{\partial x^2}+(1-l) x \frac{\partial P_0}{\partial x}\right)+\eta\rho_{13}g(z)\left(x\frac{\partial^2 P_0}{\partial x \partial r}-l \frac{ \partial P_0}{\partial r}\right)\right).
\end{equation}
\begin{prop}
The correction term $\sqrt{\delta}P_{0,1}$ is given by
\begin{equation}\label{p2}
\begin{split}
\sqrt{\delta}P_{0,1}&=V_1^\delta\frac{(T-t)^2}{2}\left(x^2\frac{\partial^2P_0}{\partial
x^2} +(1-l) x \frac{\partial P_0}{\partial x}\right)+V_2^\delta\frac{1}{\beta}\bigg[x \frac{\partial^2 P_0}{\partial \alpha \partial x}-l\,\frac{\partial
P_0}{\partial \alpha}\\
&+\frac{(T-t)^2}{2}\left(x^2\frac{\partial^2P_0}{\partial
x^2}-l\,x\frac{\partial P_0}{\partial
x}+l\,P_0\right)-(T-t)\left(x \frac{\partial^2
P_0}{\partial r \partial x}-l\,\frac{\partial P_0}{\partial r}\right)\bigg],
\end{split}
\end{equation}
in which
\begin{equation}
V_1^\delta=\sqrt{\delta}\bar{\lambda}'(z)\bar{\sigma}_1\rho_3g(z), \quad
V_2^\delta=\sqrt{\delta}\bar{\lambda}'(z)\eta\rho_{13}g(z).
\end{equation}
\end{prop}
\begin{proof}
We construct the solution from the following observations and superposition since $\langle \L_2 \rangle$ is linear:

\noindent 1)
We first observe that
$\frac{(T-t)^2}{2}(x^n\frac{\partial^n}{\partial x^n})P_0$ solves
\begin{equation}
\langle \L_2 \rangle u=-(T-t)\left(x^n\frac{\partial^n}{\partial
x^n}\right)P_0, \quad 
u(T,x,r;z)=0.
\end{equation}

\noindent 2) Next, we apply $\langle \L_2 \rangle$ on $(T-t)\frac{\partial
P_0}{\partial r}$ and obtain
\begin{equation}
\langle \L_2 \rangle\left((T-t)\frac{\partial P_0}{\partial
r}\right)=-\frac{\partial P_0}{\partial r}+(T-t)\left(-x\frac{\partial P_0
}{\partial x}+\beta\frac{\partial P_0 }{\partial r}+P_0\right),
\end{equation}
as a result of which we see that
\begin{equation}
\frac{1}{\beta}\left[-\frac{\partial P_0}{\partial
\alpha}-\frac{(T-t)^2}{2}\left(x\frac{\partial P_0}{\partial
x}-P_0\right)+(T-t)\frac{\partial P_0}{\partial r}\right]
\end{equation}
solves
\begin{equation}
\langle \L_2\rangle u=(T-t)\frac{\partial P_0}{\partial r}, \quad
u(T,x,r;z)=0.
\end{equation}
\end{proof}

\section{Calibration of the Model}\label{sec:calibration}

 In this section, we will calibrate the loss rate $l$ and the parameters 
\[
\{\bar{\lambda}, V_1^\epsilon,V_2^\epsilon,V_3^\epsilon,V_4^\epsilon,V_5^\epsilon,V_6^\epsilon,V_1^\delta,V_2^\delta\},
\]
which appear in the expressions (\ref{gp}), (\ref{p1}), and (\ref{p2}) on a daily basis (see, e.g., \cite{ronnie-timescale} and \cite{papa} for similar calibration exercises carried out only for the option data or only for the bond data).
We demonstrate this calibration on Ford Motor
Company. Note that there are some common parameters between equity options and corporate bonds. Therefore, our model will be calibrated simultaneously to both of these data sets. We will also calibrate the parameters of the interest rate and stock models to the yield curve data, historical spot rate data and historical stock price data.
Next, we test our model by using the estimated parameters to
construct an out-of-sample  CDS spread time series (3 year and 5 year), which
matches real quoted CDS spread data over the time
period ($1/6/2006-6/8/2007$) quite well. 

We also look at how our model-implied volatility matches the real option implied volatility.
We compare our results against those of \cite{ronnie-timescale}. We see that even when we make the
unrealistic assumption of constant volatility, our model is able to produce a very good fit. 

Finally, in the context of index options (when $\lambda=0$), using SPX 500 index
options data,  we show the importance of accounting for stochastic interest rates by comparing our model to that of \cite{sircar, ronnie-timescale}.

\subsection{Data description}\label{sec:data-desc}
\begin{itemize}
\item The daily closing stock price data is obtained from finance.yahoo.com.

\item The stock option data is from OptionMetrics under WRDS database, which is the same database used in \cite{carr-wu}.
\begin{itemize}
\item For index options, SPX 500 in our case, we use the data from their
Volatility Surface file. The file contains
information on standardized options, both calls and puts, with
expirations of 30, 60, 91, 122, 152, 182, 273, 365, 547, and 730
calender days. Implied volatilities there are interpolated data
using a methodology based on kernel smoothing algorithm. The
interpolated implied volatilities are very close to real data
because there are a great number of options each day for SPX 500 with
different maturities and strikes. The calibration results for index options are presented in Figure~\ref{fig:spx} and only the data set on the June 8, 2007 is used.

\item On September 15, 2006 (Friday) Ford announced that it would not be paying dividends (see e.g., {http://money.cnn.com/2006/09/15/news/companies/ford/index.htm}).  Therefore, call options on Ford after that date do not have early exercise premium starting from Sep 18, 2006. We use Ford's implied volatility surface data from 9/18/2006 to 6/8/2007. 
when we consider Ford Motor Company's options. We excluded the
observations with zero trading volume or with maturity less than 9 days.
The calibration results are used to construct Figures~\ref{fig:cds3yr}-\ref{fig:impvfouque}. In particular, the implied volatility surface data from 9/18/2006 to 6/8/2007 is used to construct \ref{fig:cds5yr}. 

\hspace{0.1in}As opposed to the options on the index there are not as many 
individual company options; and we find that the results given by
using interpolated implied volatilities in the Volatility Surface File and data implied volatilities differ.
This may be due to the fact that there are a limited number of option prices available for individual
companies; i.e., there may not be enough data points for the implied volatilities to be accurately
interpolated. Therefore, we use the Option Price file, which contains the historical option price information, of the OptionMetrics database \end{itemize}

\item For each day we U.S government Treasury yield data with maturities: 1 month, 3 monts, 6 months, 1 year, 2 years, 3 years, 5 years, 7 years, 10 years, 20 years.
This data set is available at: {www.treasury.gov/offices/domestic-finance/debt-management/interest-rate/yield.shtml.}

\item Corporate bond and CDS data is obtained from Bloomberg. Number of available bond quotes and bond maturities vary. Typically there are around 15 data points, for example, on June 8th, we have the following maturities: 0.60278,1.0222,1.1861,1.3139,1.4083, \newline 1.5944, 2.3889, 2.6028, 3.0194, 3.2694, 3.3972, 3.6472, 4.1722, 4.3806, 6.3139, 9.5194.

\end{itemize}

\subsection{The parameter estimation}\label{sec:parest}

The following parameters can be directly estimated from the spot-rate and stock price historical data:
\begin{enumerate}
\item The parameters of the interest rate model $\{\alpha,\beta,\eta\}$ are obtained by a  least-square fitting to the
Treasury yield curve as in \cite{papa}.
\item $\bar{\rho}_1=\frac{\bar{\sigma}_1}{\bar{\sigma}_2} \rho_1$, the ``effective" correlation between risk-free
spot rate $r$ (we use the one-month treasury bonds as a proxy for $r$) and stock price in (\ref{eq:effective}) is estimated from historical risk-free spot rate and
stock price data.
\item $\bar{\sigma}_2$, the ``effective" stock price volatility in (\ref{eq:effective})
is estimated from the historical stock price data.
\end{enumerate}

Now, we detail the calibration method for $l$, $\bar{\lambda}(z)$ and
$\{V_1^\epsilon,V_2^\epsilon,V_3^\epsilon,V_4^\epsilon,V_5^\epsilon,V_6^\epsilon,V_1^\delta,V_2^\delta\}$. We will minimize the in-sample quadratic pricing error using non-linear least squares to calibrate these parameters on a daily basis. This way we find a risk neutral model that matches a set of observed market prices. This risk neutral model can then be used to price more exotic, illiquid or over-the-counter derivatives. This practice is commonly employed; and for further discussion of this calibration methodology we refer to \cite{cont} (see Chapter 13 and the references therein).

Our  calibration is carried out in two steps in tandem:

\noindent \textbf{Step 1. Estimation of $l\bar{\lambda}$ and
$\{lV_3^\epsilon,lV_2^\delta\}$ from the corporate bond price data.}

The approximate price formula in (\ref{eq:apprx}) for a defaultable bond is
\begin{equation}\label{eq:apprx-bndprice}
\widetilde{B}^c=B_0^c+\sqrt{\eps}B^c_{1,0}+\sqrt{\delta}B^c_{0,1},
\end{equation}
in which $B_0^c$ is given by (\ref{btt}) and
\begin{equation}
\begin{split}
\sqrt{\epsilon}B^c_{1,0}&=lV_3^\epsilon \frac{\partial B^c_0}{\partial \alpha},%=-lV_3^\epsilon\frac{1}{\beta}\left((T-t)B^c_0+\frac{\partial B^c_0}{\partial r}\right)\\
\\ \sqrt{\delta}B^c_{0,1}&=lV_2^\delta\frac{1}{\beta}\left[-\frac{\partial
B^c_0}{\partial \alpha}
+\frac{(T-t)^2}{2}B^c_0+(T-t)\frac{\partial
B^c_0}{\partial r}\right].
%\\&=lV_2^\delta\frac{1}{\beta}\left[\frac{1}{\beta}\left((T-t)B^c_0+\frac{\partial
%B^c_0}{\partial r}\right) +\frac{(T-t)^2}{2}B^c_0+(T-t)\frac{\partial
%B^c_0}{\partial r}\right].
\end{split}
\end{equation}

 We obtain $\{l \bar{\lambda}(z),l V_{3}^{\eps}, l V_{2}^{\delta}\}$ from least-squares fitting, i.e. by minimizing
\begin{equation}\label{eq:bond-lls}
\displaystyle\sum_{i=1}^n(B^c_{\text{obs}}(t,S_i)-B^c_{\text{model}}(t,S_i;l\bar{\lambda},lV_3^\epsilon,lV_2^\delta))^2,
\end{equation}
where $B^c_{\text{obs}}(t,S_i)$ is the observed market price of a
bond that matures at time $S_i$ and $B^c_{\text{model}}(t,S_i;l\bar{\lambda},lV_3^\epsilon,lV_2^\delta)$ is the corresponding model price obtained from (\ref{eq:apprx-bndprice}). Here, $n$ is the number of bonds that are traded at time $t$.
For a fixed value of $l \bar{\lambda}(z)$
 it follows from (\ref{eq:apprx-bndprice}) that $\{l V_{3}^{\eps}, l V_{2}^{\delta}\}$ can be determined as the least squares solution of
\begin{equation*}
\begin{pmatrix}
\frac{\partial B^c_0}{\partial \alpha}(t,S_1),&
\frac{1}{\beta}\left[-\frac{\partial B^c_0}{\partial \alpha}
+\frac{(S_1-t)^2}{2}B^c_0+(S_1-t)\frac{\partial B^c_0}{\partial r}\right] 
\\ \vdots & \vdots
\\ \frac{\partial B^c_0}{\partial \alpha}(t,S_n),&
\frac{1}{\beta}\left[-\frac{\partial B^c_0}{\partial \alpha}
+\frac{(S_n-t)^2}{2}B^c_0+(S_n-t)\frac{\partial B^c_0}{\partial r}\right] 
\end{pmatrix}
\begin{pmatrix}
lV_3^\epsilon\\
lV_2^\delta
\end{pmatrix}=
\begin{pmatrix}
B^c_{\text{obs}}(t,S_1)-B_0^c(t,S_1;l\bar\lambda)
\\ \vdots 
\\ B^c_{\text{obs}}(t,S_n)-B_0^c(t,S_n;l\bar\lambda)
\end{pmatrix}.
\end{equation*}
Now, we vary $l\bar{\lambda}(z) \in [0,M_1]$ and choose the point
$\{l\bar{\lambda},lV_3^\epsilon,lV_2^\delta\}$ that minimizes (\ref{eq:bond-lls}). Here, we take $M_1=1$ guided by the results of \cite{papa}.

\noindent \textbf{Step 2. Estimation of
$\{l, V_1^\epsilon,V_2^\epsilon,V_4^\epsilon,V_5^\epsilon,V_6^\epsilon,V_1^\delta\}$ from the equity option data}:
These parameters are calibrated from the stock options data by a least-squares fit to the observed implied volatility. We choose the parameters to minimize
\begin{equation}\label{eq:impliedvolfit}
\begin{split}
&\displaystyle\sum_{i=1}^n(I_{\text{obs}}(t,T_i,K_i)-I_{\text{model}}(t,T_i,K_i;\text{model
parameters} ))^2\\
&\hspace{1.6in} \approx\displaystyle\sum_{i=1}^n\frac{(P_{\text{obs}}(t,T_i,K_i)-P_{\text{model}}(t,T_i,K_i;\text{model
parameters} ))^2}{\text{vega}^2(T_i,K_i)}
\end{split}
\end{equation}
in which $I_{\text{obs}}(t,T_i,K_i)$ and
$I_{\text{model}}(t,T_i,K_i;\text{model parameters})$ are observed
Black-Scholes implied volatility and model Black-Scholes implied
volatility, respectively. The right hand side of (\ref{eq:impliedvolfit}) is from \cite{cont}, page 439.
Here,
$P_{\text{obs}}(t,T_i,K_i)$ is the market
price of a European option (a put or a call) that matures at time
$T_i$ and with strike price $K_i$ and
$P_{\text{model}}(t,T_i,K_i;\text{model parameters})$ is the
corresponding model price which is obtained from (\ref{eq:apprx}). As in \cite{cont}, $\text{vega}(T_i,K_i)$ is the market implied
Black-Scholes vega. 

Let $P_0(t,T_i,K_i;\bar{\lambda}(z))$ be either of (\ref{eq:call-opt}) and (\ref{eq:put-opt}) with $K=K_i$ and $T=T_i$.
Let us introduce the Greeks,
\begin{equation}\label{eq:greeks}
\begin{split}
g_1&=-(T-t)x^2\frac{\partial^2P_0}{\partial
x^2},
\quad g_2=-(T-t) x\frac{\partial}{\partial
x}\left(x^2\frac{\partial^2P_0}{\partial x^2}\right), \quad
g_3=\frac{\partial}{\partial
\alpha}\left(x\frac{\partial P_0}{\partial x}-P_0\right),\\
g_4&=x^2\frac{\partial^3P_0}{\partial x^2 \partial \alpha}, \quad
g_5=x\frac{\partial^2 P_0}{\partial \eta \partial x}, \quad
g_6=x\frac{\partial^2 P_0}{\partial
\alpha \partial x},\quad
g_7=\frac{(T-t)^2}{2}x^2\frac{\partial^2P_0}{\partial x^2},\\
g_8&=\frac{1}{\beta}\left[x\frac{\partial^2
P_0}{\partial \alpha \partial x}-\frac{\partial P_0}{\partial \alpha}
+\frac{(T-t)^2}{2}\left(x^2\frac{\partial^2P_0}{\partial
x^2}-x\frac{\partial P_0}{\partial
x}+P_0\right)-(T-t)\left(x(\frac{\partial^2
P_0}{\partial r} \partial x)-\frac{\partial P_0}{\partial r}\right)\right],
\end{split}
\end{equation}
in which each term can be explicitly evaluated (see Appendix). 

Now from (\ref{eq:apprx}) and the results of Section~\ref{sec:explct} (with $l=1$), we can write
\begin{equation}\label{eq:7param-model}
\begin{split}
P_{\text{model}}(t,T_i,K_i;\text{model
parameters})
&=P_0(t,T_i,K_i;\bar\lambda)
+V_1^\epsilon g_1(T_i,K_i;\bar\lambda)+V_2^\epsilon
g_2(T_i,K_i;\bar\lambda)\\&+V_3^\epsilon
g_3(T_i,K_i;\bar\lambda)
+V_4^\epsilon g_4(T_i,K_i;\bar\lambda)+V_5^\epsilon
g_5(T_i,K_i;\bar\lambda)
\\&+V_6^\epsilon g_6(T_i,K_i;\bar\lambda)
+V_1^\delta g_7(T_i,K_i;\bar\lambda)+V_2^\delta
g_8(T_i,K_i;\bar\lambda).
\end{split}
\end{equation}
First, let us fix the value of $l$. Then, from Step 1, we can infer the values of $\{\bar\lambda, V_3^\epsilon,V_2^\delta\}$. Now
the fitting problem in (\ref{eq:impliedvolfit}) is a linear least squares problem for
$\{V_1^\epsilon,V_2^\epsilon,V_4^\epsilon,V_5^\epsilon,V_6^\epsilon,V_1^\delta\}$.
Next, we vary $l \in [0,1]$ and choose 
$\{l,V_1^\epsilon,V_2^\epsilon,V_4^\epsilon,V_5^\epsilon,V_6^\epsilon,V_1^\delta\}$ so that (\ref{eq:impliedvolfit}) is minimized. 

\subsection{Model implied CDS spread matches the observed CDS spread}

Let $\widetilde{B}^c(t,T;l)$ denote the approximation for the
price at time $t$ of a defaultable bond that matures at
time $T$, and has loss rate $l$ (see (\ref{eq:apprx-bndprice})). Let $B(t,T)$ be the price
of a risk-free bond. Then, the model implied CDS spead with maturity $T_M$ is
\begin{equation}\label{eq:imcdssp}
c^{ds}_{\text{model}}(t,T_M)=\frac{B(t,T_M)-\widetilde{B}^c(t,T_M;l)}{\displaystyle\sum_{m=1}^M
\widetilde{B}^c(t,T_m;1)}.
\end{equation}
Recall that we have already estimated all of the model parameters in Section~\ref{sec:parest} using both corporate term structure data and the stock option implied volatility surface. Therefore, using  (\ref{eq:imcdssp}) we can plot the model implied CDS spread over time and compare it with the CDS spread data available in the market. This is precisely what we do in Figures~\ref{fig:cds3yr} and \ref{fig:cds5yr}. We look at the time series $c_{\text{model}}^{ds}(t,3)$ and $c_{\text{model}}^{ds}(t,5)$ and compare them to the CDS spread time series of the Ford Motor Company. The match seems to be extremely good, which attests to the power of our modeling framework.

By varying $T_M$ in (\ref{eq:imcdssp}) we can obtain the model implied term structure of the CDS spread. Figure~\ref{fig:cdssamples} shows the range of shapes we can produce.

\subsection{Fitting Ford's implied volatility}\label{sec:Fordsimplied}

We will compare how well our model fits the implied volatility against the model of \cite{ronnie-timescale}, which does not account for the default risk and for the randomness of the interest rates.
Although, we only calibrate seven parameters (hence we refer to our model as the 7-parameter model) to the option prices (see the second step of the estimation in Section~\ref{sec:parest}), we have many more parameters than the model of \cite{ronnie-timescale}, which only has four parameters (we refer to this model as the 4-parameter model). Therefore, for a fair comparison, we also consider a model in which the volatility is a constant. In this case, as we shall see below, there are only three parameters to calibrate to the option prices, therefore we call it the 3-parameter model.

\noindent \textbf{Constant Volatility Model}
 In this case, we take $\bar{\sigma}_1=\bar{\sigma}_2=\sigma$ in the expression for $P_0$ in Corollary~\ref{cor:p0} .
The expression for $\sqrt{\delta}P_{0,1}$ remains the same as before. However, $\sqrt{\epsilon}
P_{1,0}$ simplifies to
\begin{gather}
\sqrt{\epsilon} P_{1,0}=-(T-t)V_1^\epsilon
x^2\frac{\partial^2P_0}{\partial x^2}+V_3^\epsilon
\left(-x \frac{\partial^2 P_0}{\partial
\alpha \partial x}+\frac{\partial P_0}{\partial \alpha}\right).
\end{gather}
This model has only three parameters, $l, V_1^{\eps}, V_{1}^{\delta}$ that need to be calibrated to the options prices, as opposed to the 4-parameter model of \cite{ronnie-timescale}.

As it can be seen from Figure~\ref{fig:impvolFord} as expected our 7-parameter model outperforms the 
4-parameter model of \cite{ronnie-timescale} as expected and fits the implied volatility data well. But, what is surprising is that the 3-parameter model, which does not account for the volatility but accounts for the default risk and stochastic interest rate, has almost the same performance as the 7-parameter model.

The 7-parameter model has a very rich implied volatility surface structure, the surface has more curvature than that of the 4-parameter model of \cite{ronnie-timescale}, whose volatility  surface is more flat; see Figures~\ref{fig:imp7p} and \ref{fig:impvfouque}. (The parameters to draw these figures are obtained by calibrating the models to the data implied volatility surface on June 8 2007.)
The 7-parameter model has a recognizable skew even for longer maturities and has a much sharper skew for shorter maturities.

\subsection{Fitting the implied volatility of the index options}\label{sec:demnst}
The purpose of this section is to show the importance of accounting for stochastic interest rates in fitting the implied volatility surface. Interest rate changes should, indeed, be accounted for in pricing long maturity options. When we price index options, we set $\bar{\lambda}=0$ and our approximation in (\ref{eq:apprx})
simplifies to
\begin{equation}\label{eq:indx}
P^{\epsilon,\delta}\approx P_0+\sqrt{\epsilon} P_{1,0},
\end{equation}
in which 
$P_0$ is given by Corollary~\ref{cor:p0} after settiing $\bar{\lambda}(z)=0$, and
\begin{equation}
\sqrt{\epsilon}P_{1,0}=-(T-t)\left(V_1^\epsilon
x^2\frac{\partial^2P_0}{\partial x^2}+V_2^\epsilon
x\frac{\partial}{\partial x}\left(x^2\frac{\partial^2P_0}{\partial
x^2}\right)\right)+V_4^\epsilon x^2\frac{\partial^3P_0}{\partial x^2 \partial \alpha}+V_5^\epsilon
x \frac{\partial^2 P_0}{\partial
\eta \partial x}+V_6^\epsilon x\frac{\partial^2
P_0}{\partial \alpha \partial x}
\end{equation}
Note that the difference of (\ref{eq:indx}) with the model of \cite{ronnie-timescale} is that the latter allows for a slow evolving volatility factor to better match the implied volatility at the longer maturities. This was an improvement on the model of \cite{sircar}, which only has a fast scale component in the volatility model.
We, on the other hand, by accounting for stochastic interest rates, capture the same performance by using only a fast scale volatility model.

From Figure~\ref{fig:spx}, we see that both (\ref{eq:indx}) and \cite{ronnie-timescale} outperform the model of \cite{sircar}, especially at the longer maturities ($T=$ 9 months, 1 year, 1.5 years and 2 years), and their performance is very similar. This observation emphasizes the importance of accounting for  stochastic interest rates for long maturity contracts. 

\section*{Appendix: Explicit formulae for the Greeks  in (\ref{eq:greeks})}
When $h(x)=(x-K)^+$, we can explicitly express the Greeks in (\ref{eq:greeks}) 
in terms of
$
f(x)=\frac{1}{\sqrt{2\pi}}\exp(-x^2/2)
$ 
as
\begin{align*}
&x^2\frac{\partial^2C_0}{\partial
x^2}=\frac{xf(d_1)}{\sqrt{v_{t,T}}},
\quad x\frac{\partial}{\partial
x}\left(x^2\frac{\partial^2C_0}{\partial
x^2}\right)=\frac{xf(d_1)}{\sqrt{v_{t,T}}}\left(1-\frac{d_1}{\sqrt{v_{t,T}}}\right),\\
&\frac{\partial}{\partial \alpha}\left(x\frac{\partial
C_0}{\partial x}-C_0\right)=-K\bar
B^c(t,T)\left(\frac{T-t}{\beta}+\frac{\exp(-\beta(T-t))-1}{\beta^2}\right)\left(N(d_2)-\frac{f(d_2)}{\sqrt{v_{t,T}}}\right),\\
&\frac{\partial}{\partial
\alpha}\left(x^2\frac{\partial^2C_0}{\partial
x^2}\right)=\frac{-xf(d_1)d_1}{v_{t,T}}\left(\frac{T-t}{\beta}+\frac{\exp(-\beta(T-t))-1}{\beta^2}\right),\\
&x\frac{\partial}{\partial x}\left(\frac{\partial C_0}{\partial
\alpha}\right)=\frac{xf(d_1)}{\sqrt{v_{t,T}}}\left(\frac{T-t}{\beta}+\frac{\exp(-\beta(T-t))-1}{\beta^2}\right),\\
&\frac{\partial}{\partial r}\left(x\frac{\partial C_0}{\partial
x}-C_0\right)=-K\bar
B^c(t,T)\left(\frac{1-\exp(-\beta(T-t))}{\beta}\right)\left(N(d_2)-\frac{f(d_2)}{\sqrt{v_{t,T}}}\right),\\
&x\frac{\partial}{\partial x}\left(\frac{\partial C_0}{\partial
\eta}\right)=xf(d_1)\bigg[-\frac{1}{\sqrt{v_{t,T}}}\bigg(\frac{\eta}{\beta^2}(T-t)+\frac{2\eta}{2\beta^3}(\exp(-\beta(T-t))-1\bigg)-\frac{\eta}{2\beta^3}(\exp(-2\beta(T-t))-1))\\
&+\bigg(-\frac{1}{2}\log\left(\frac{x}{K\bar{B}_{t,T}^c}v_{t,T}^{-3/2}+\frac{1}{4\sqrt{v_{t,T}}}\right)\bigg) \times \\
&\bigg(\left(\frac{2\bar\rho_1\bar\sigma_2}{\beta}+\frac{2\eta}{\beta^2}\right)(T-t)+\left(\frac{2\bar\rho_1\bar\sigma_2}{\beta^2}+\frac{4\eta}{\beta^3}\right)\exp(-\beta(T-t))-\frac{\eta}{\beta^3}\exp(-2\beta(T-t))-\left(\frac{2\bar\rho_1\bar\sigma_2}{\beta^2}+\frac{3\eta}{\beta^3}\right)\bigg)\bigg].
\end{align*}

\bibliographystyle{dcu}
{\small
\bibliography{references}}

\newpage

\begin{figure}[ht]
\begin{center}
\includegraphics[width = 1\textwidth,height=5.8cm]{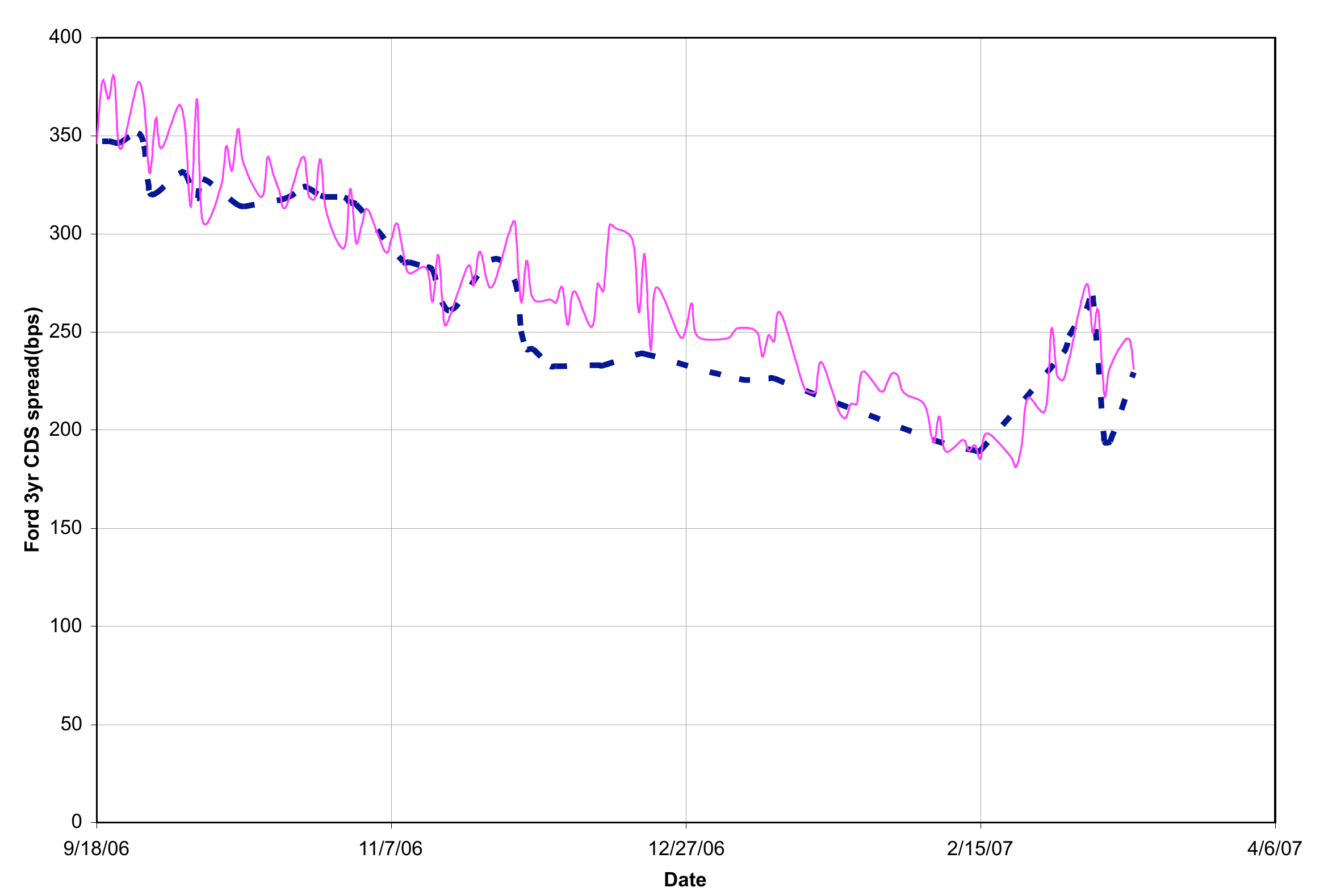}
\caption{Ford 3 year CDS annual spread time series from 9/18/2006-3/13/2007. \newline Spread implied by model is pink
solid line, real quoted spread is blue broken line. Ford's 3 year CDS spread time series is not available from 3/13/2007 until 8/31/07 in our data source.}
\label{fig:cds3yr}
\end{center}
\end{figure}

\begin{figure}[hb]
\begin{center}
\includegraphics[width = 1\textwidth,height=6cm]{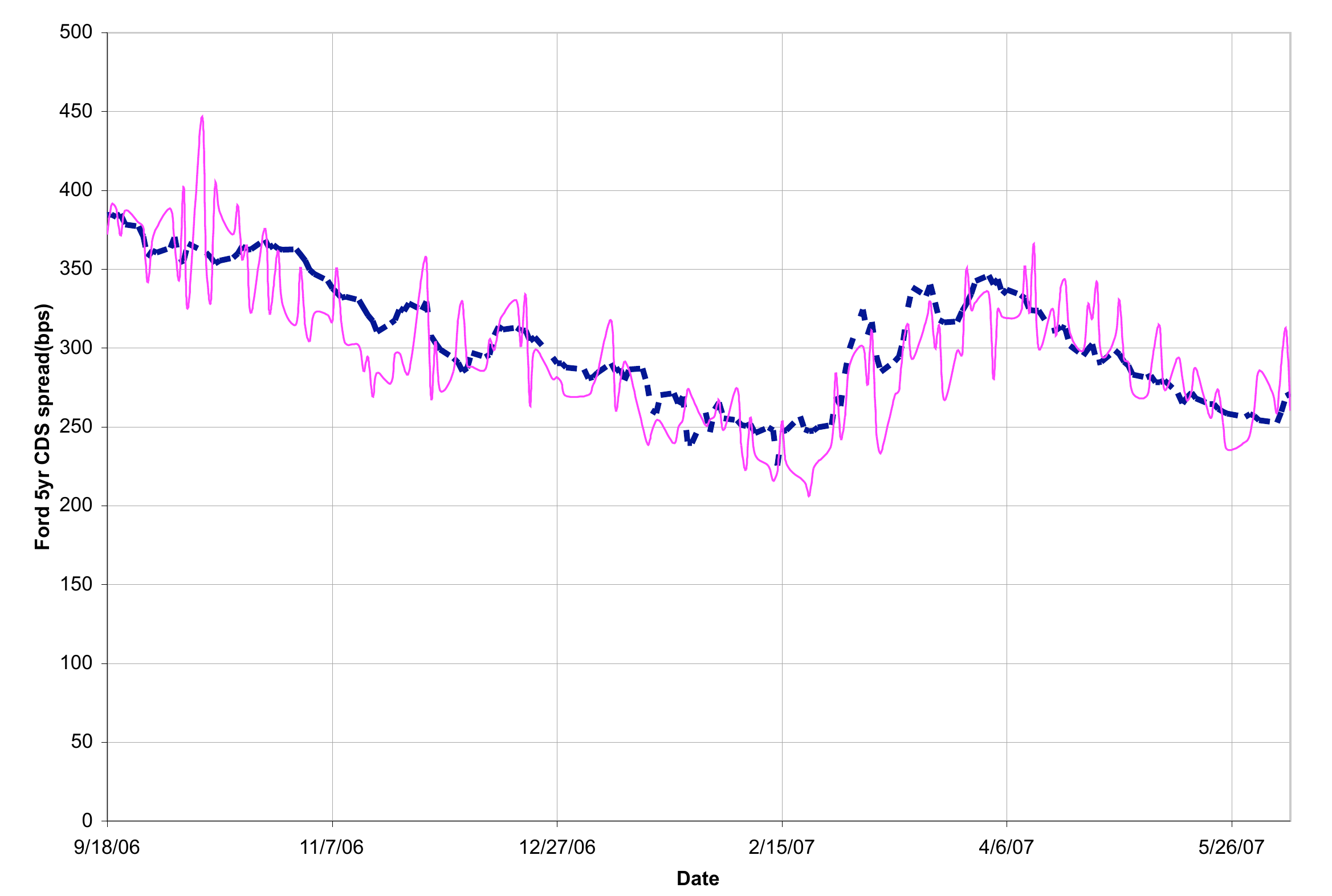}
\caption{Ford 5 year CDS annual spread time series from 9/18/2006-6/8/2007. \newline Spread implied by model is pink
solid line, real quoted spread is blue broken line.}
\label{fig:cds5yr}
\end{center}
\end{figure}

\begin{figure}[hb]
\begin{center}
\includegraphics[width = 1\textwidth,height=10cm]{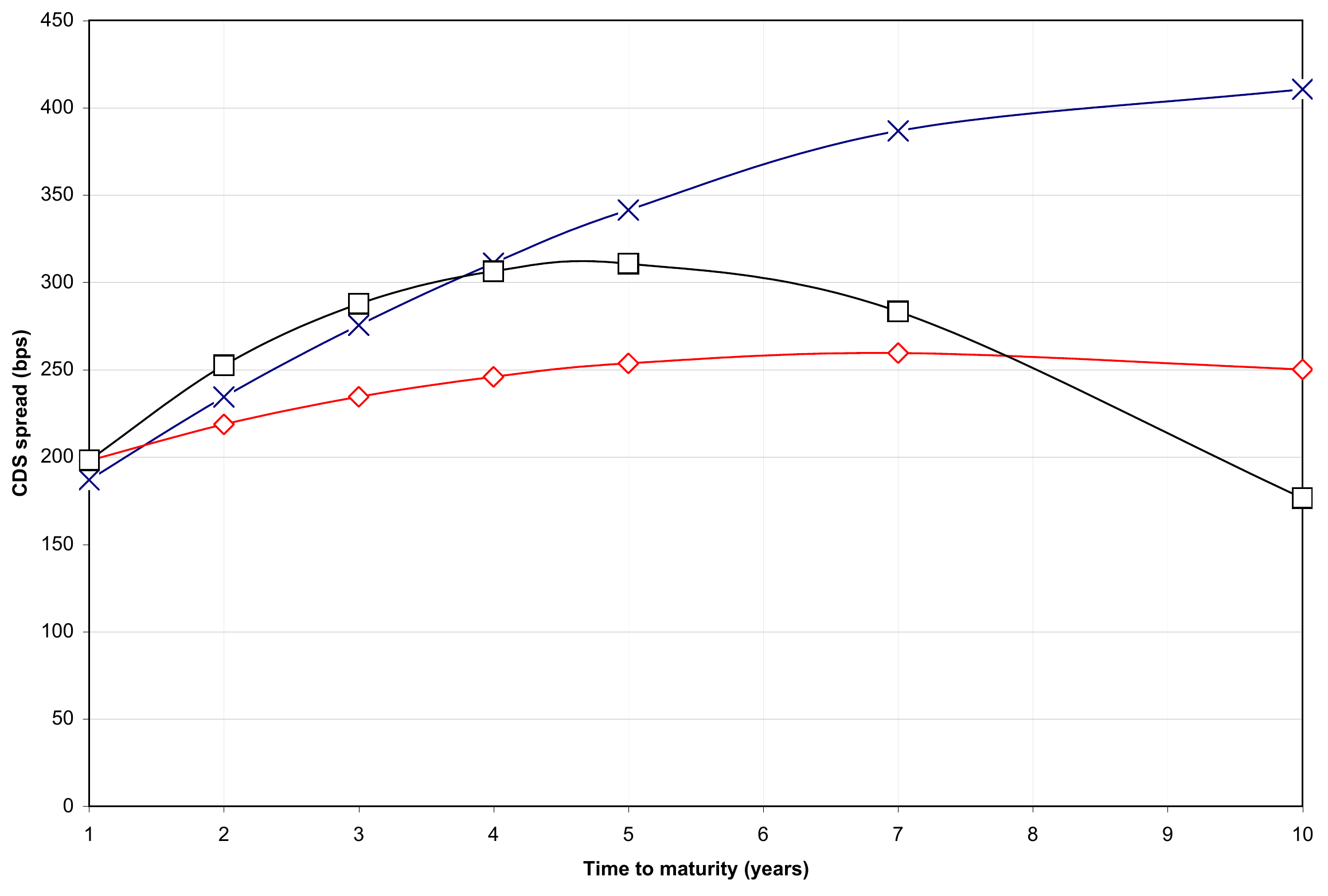}
\caption{CDS Term Structures (\ref{eq:imcdssp}) can produce:
\newline \textbf{Legend} \newline
-x-, blue (The parameters are obtained from calibration to 11/13/2006): 
   $\alpha$=0.0037, $\beta$=0.0872
   $\eta=0.0001$,
   $r=0.0516$,
   $l(\text{loss rate})=0.283$,
   $\bar{\lambda}(z)=0.0459$,
   $[V_3^\epsilon,V_2^\delta]= [0.0425,0.0036]$. \newline \newline
   -squares-, black (The parameters correspond to 6/18/2006):
  $ \alpha=0.0045$, $\beta=0.0983$,
   $\eta=0.0002$,
   $r=0.0516$,
   $l=1$, 
   $\bar{\lambda}=0.012$,
   $[V_3^\epsilon,V_2^\delta]=[0.0185,0.0025]$,
  \newline \newline
   -diamonds-, red (The parameters correspond to 9/22/2006):
   $\alpha=0.0039$,  $\beta=0.0817$,
   $\eta=0.0012$, 
   $r=0.0496$,
   $l$=1,
   $\bar{\lambda}(z)=0.017$,
   $[V_3^\epsilon,V_2^\delta]=[0.0067,0.0005]$
}
\label{fig:cdssamples}
\end{center}
\end{figure}

\begin{figure}[hb]
\begin{center}
\includegraphics[width = 1\textwidth,height=5.5in]{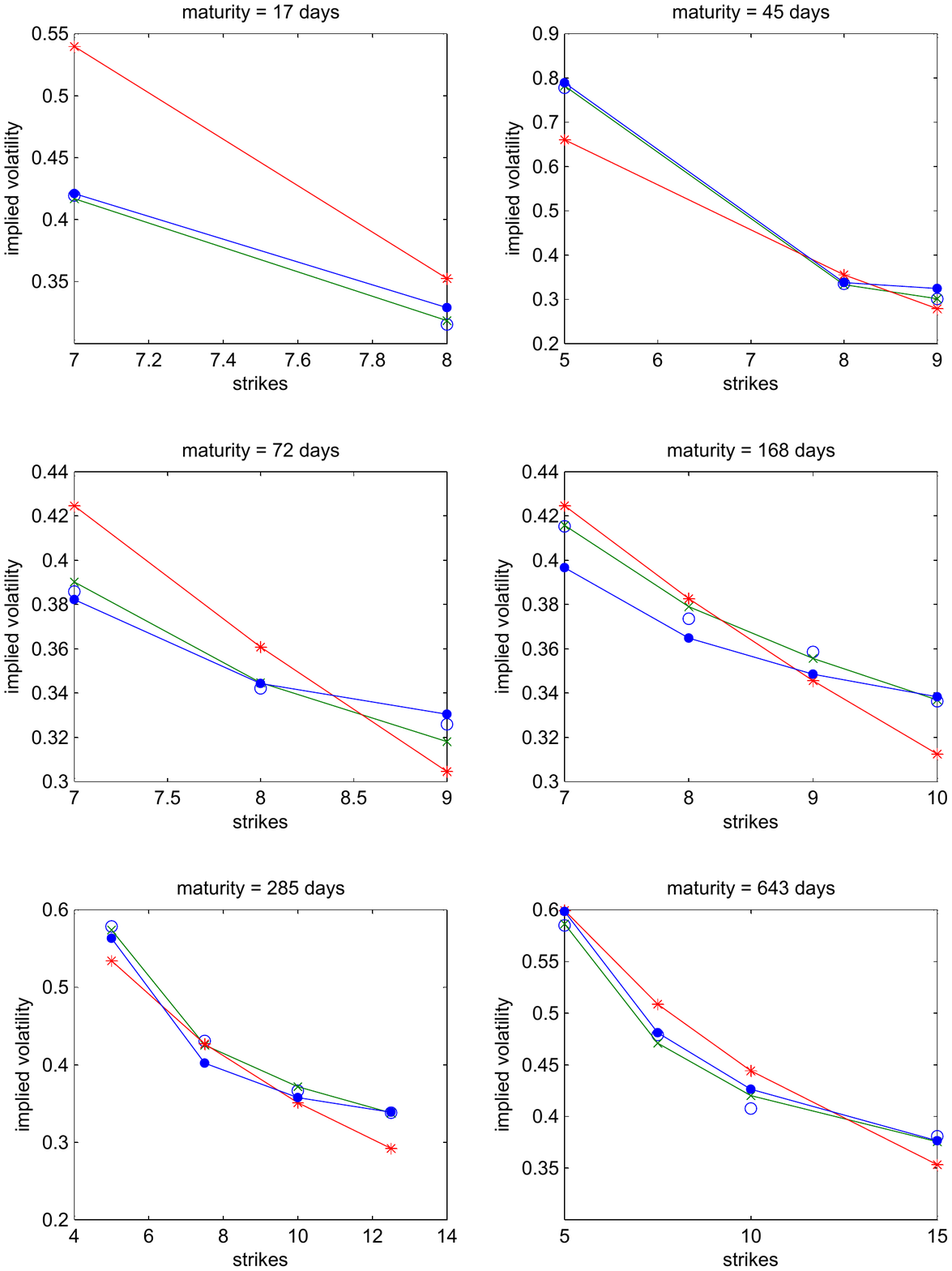}
\caption{
Implied volatility fit to the Ford call option data with maturities of [17,45,72,168,285,643] calender days on  April 4, 2007. \newline Model is calibrated aross all maturities but  we plotted the implied volatilities for each maturity, separately.
Here, stock price $(x)=8.04$, historical volatility $(\bar{\sigma}_2)=0.3827$, one month treasury rate $(r)=0.0516$, 
estimated correlation between risk-free spot rate(one month treasury) and stock price $(\bar{\rho}_1)=-0.0327$. 
Also $\alpha=0.0037$, $\beta=0.0872$, $\eta=0.0001$
which are obtained with a least-square fitting to the Treasury
yield curve on the 4th of April.
\newline
\textbf{Legend:} \newline
'o', empty circles = observed data; \newline 
'x', green = stochastic vol+stochastic hazard rate+stochastic interest rate = the 7-parameter model; \newline
small full circle, blue = constant vol+stochastic hazard rate+ stochastic interest rate = the 3-parameter model \newline
'*', red = The model of \cite{ronnie-timescale} which has constant interest rate+stochastic vol (slow and fast scales) = the 4 parameter model.
 }
\label{fig:impvolFord}
\end{center}
\end{figure}

\begin{figure}[hb]
\begin{center}
\includegraphics[width =0.8\textwidth,height=7.4cm]{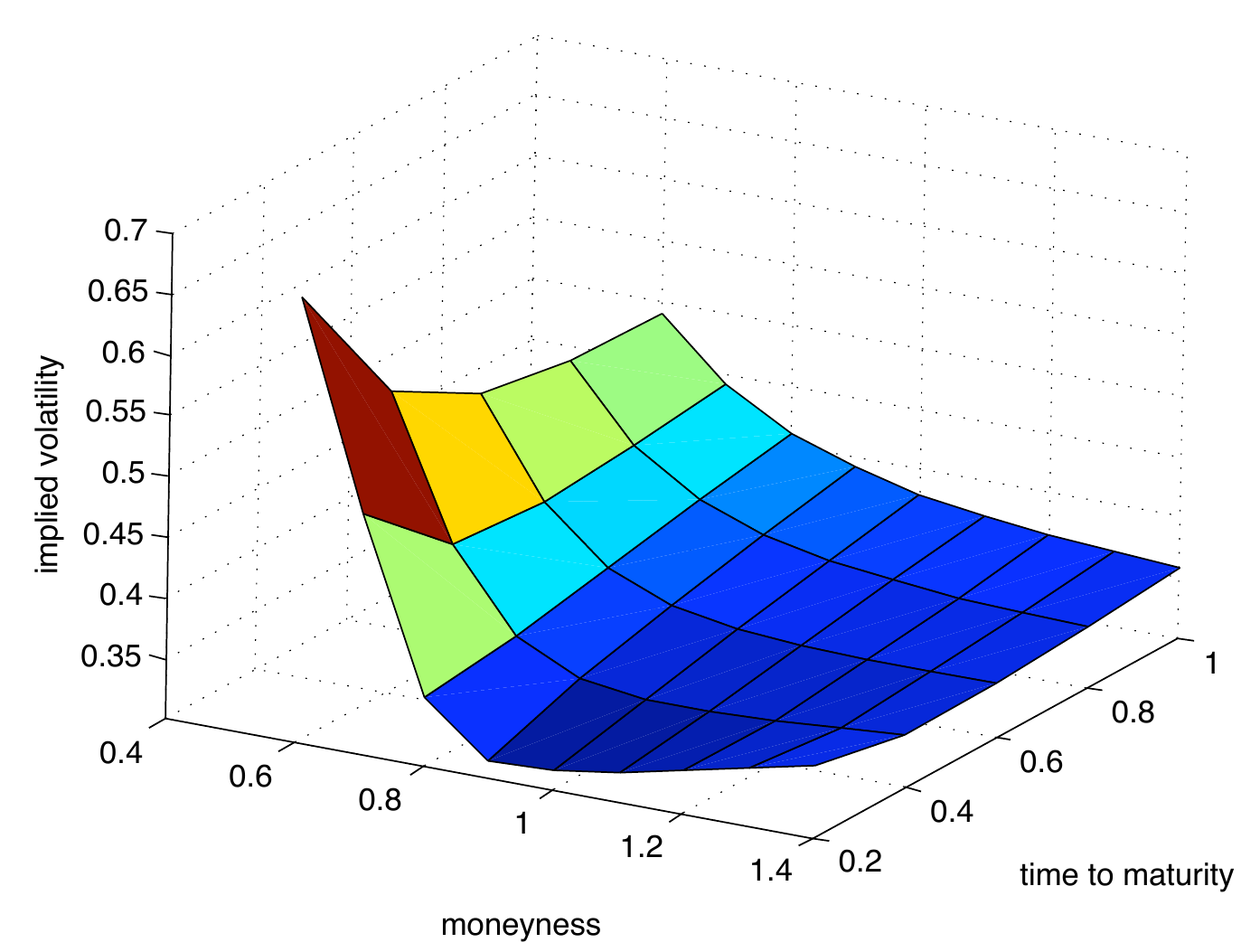}
\caption{Implied volatility surface corresponding to (\ref{eq:7param-model}), the 7-parameter model. \newline 
Here, $\alpha=0.0063$, $\beta=0.1034$, $\eta=0.012$, $r=0.0476$
 $\bar{\sigma}_2=0.2576$, $\bar{\lambda}(z)=0.027$,
 $(V_1^\epsilon,V_2^\epsilon,V_3^\epsilon,V_4^\epsilon,V_5^\epsilon,V_6^\epsilon,V_1^\delta,V_2^\delta)=(0.9960,-0.0014,0.0009,0.0104,-0.6514,0.3340,-0.1837,-0.0001)$.
}
\label{fig:imp7p}
\end{center}
\end{figure}

\begin{figure}[hb]
\begin{center}
\includegraphics[width =0.8\textwidth,height=7.4cm]{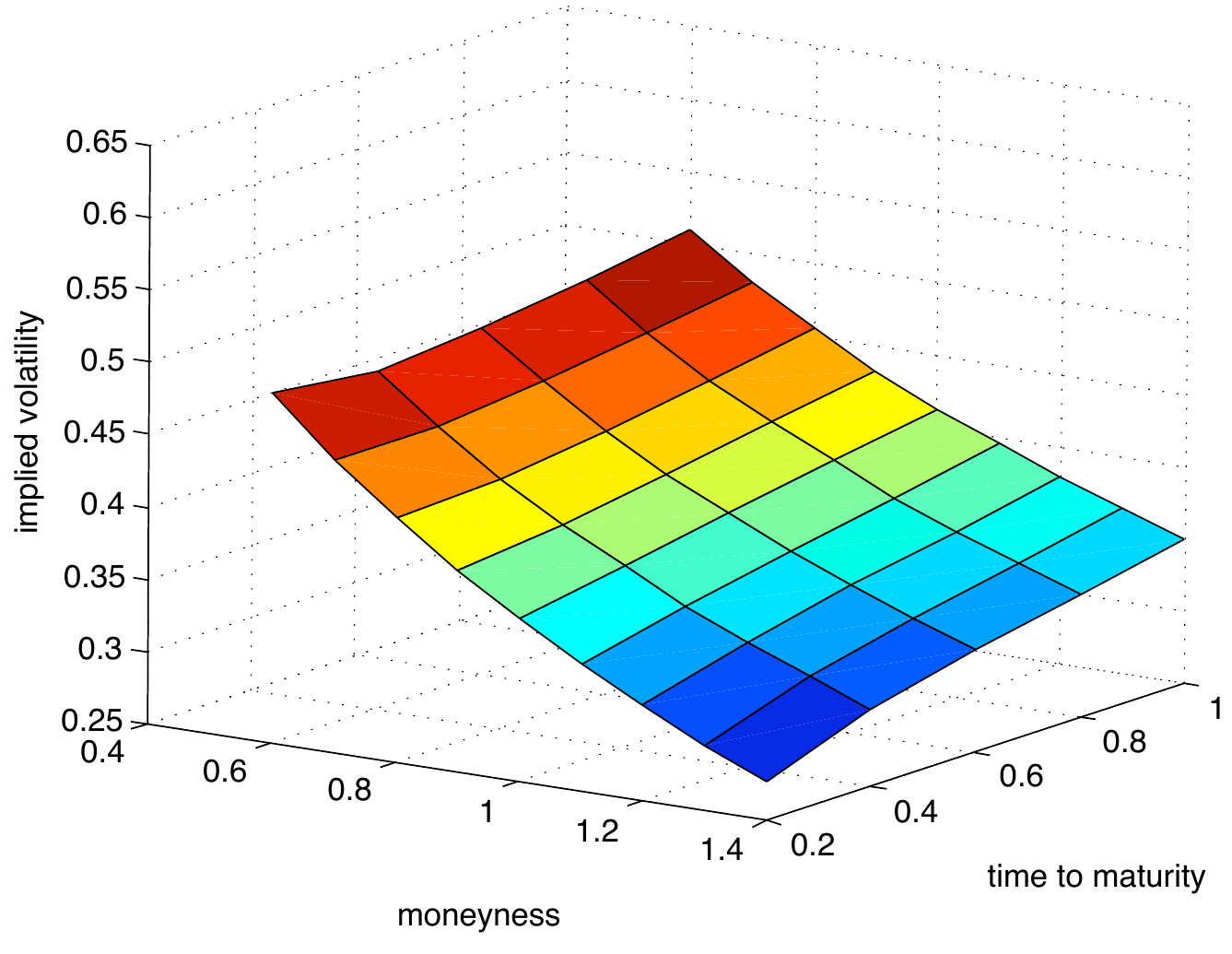}
\caption{Implied Volatility Surface corresponding to the 4-parameter model of \cite{ronnie-timescale}. \newline Here, $r=0.046$, average volatility=0.2546, and the parameters in (4.3) of \cite{ronnie-timescale} are choosen to be $(V^{\eps}_2, V^{\eps}_3,V_0^{\delta},V_1^{\delta})= (-0.0164,-0.1718,0.0006,0.0630)$. Note that the parameters here and Figure~\ref{fig:imp7p} are both obtained by calibrating the models to the data implied volatility surface of Ford Motor Company on June 8, 2007.}
\label{fig:impvfouque}
\end{center}
\end{figure}

\begin{figure}[hb]
\begin{center}
\includegraphics[width = 1\textwidth,height=6in]{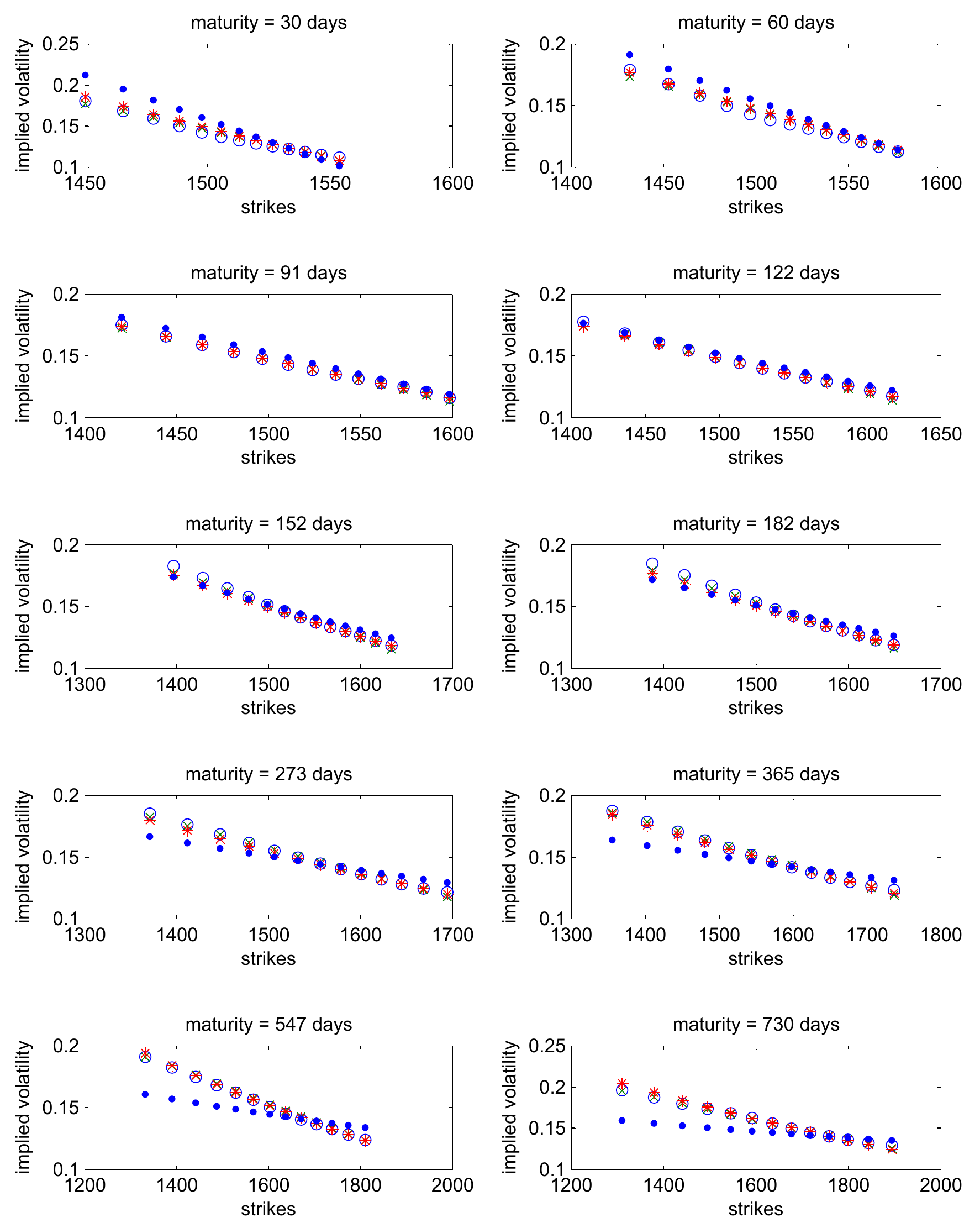}
\caption{The fit to the Implied Volatility Surface of SPX on June 8, 2007 with
maturities [30,60,91,122,152,182,273,365,547,730] calender days. Recall from Section~\ref{sec:data-desc} that we use standardized options from the OptionMetrics.
\newline
Models are calibrated aross all maturities, but we plot the implied volatility fits separately.
The parameters are:
stock price $(x)=1507.67$,
dividend rate = 0.0190422, historical volatility $(\bar{\sigma}_2)=0.1124$, one month treasury rate $(r)=0.0476$, estimated correlation between risk-free spot 
rate(one month treasury) and stock price $(\bar{\rho}_1)=0.020454$.
Also, 
$\alpha=0.0078$, 
$\beta=0.1173$,
$\eta=0.0241$, which
 are obtained from a least-square fitting to the Treasury
yield curve of the same day.
\newline
\textbf{Legend} \newline
'o', empty cirles = observed data, \newline
'x", green = Implied volatility of (\ref{eq:indx}), \newline
'*', red = Implied volatility of \cite{ronnie-timescale}, \newline
small full circle, blue = Implied volatility of \cite{sircar}.}
\label{fig:spx}
\end{center}
\end{figure}

\end{document}